\RequirePackage{amsmath}

\documentclass[runningheads]{llncs}
\usepackage[numbers]{natbib}
\usepackage[T1]{fontenc}
\usepackage{booktabs} %
\usepackage[ruled,linesnumbered]{algorithm2e} %

\usepackage{mathtools}
\usepackage{amsfonts}
\usepackage{amssymb}
\usepackage{mathrsfs}
\usepackage{hyperref}

\SetAlFnt{\small}
\SetAlCapFnt{\small}
\SetAlCapNameFnt{\small}
\SetAlCapHSkip{0pt}
\IncMargin{-\parindent}

\newcommand\ps{\texttt{PSPACE}\xspace}
\newcommand\gmbs{\texttt{GMBS}\xspace}
\newcommand\lmbs{\texttt{LMBS}\xspace}
\newcommand\stwosat{\texttt{MAX-S2SAT}\xspace}
\newcommand\sthreesat{\texttt{MAX-S3SAT}\xspace}
\newcommand\newinst{\mathcal{I}_{\phi}\xspace}
\newcommand{\lam}{\mathcal{L}}
\newcommand{\mat}{\mathcal{M}}
\newcommand{\expect}{\mathbb{E}}
\newcommand{\varx}{\mathcal{X}}
\newcommand{\vy}{\mathcal{Y}}
\newcommand{\vecs}{\vec{s}}

\newcommand{\sets}{\mathcal{S}}
\newcommand{\smallbin}{\mathscr{S}}
\newcommand{\largebin}{\mathscr{B}}
\newcommand{\covariance}{\textbf{Cov}}
\newcommand{\prob}{\textbf{Pr}}

\usepackage[margin=1.0in]{geometry}

\title{Matroid Bayesian Online Selection}

\author{Ian DeHaan\thanks{Supported by an NSERC Canada Graduate Scholarship.} \and
Kanstantsin Pashkovich\thanks{Supported by NSERC Discovery Grants Program RGPIN-2020-04346.}}

\authorrunning{I. DeHaan and K. Pashkovich}
\institute{Department of Combinatorics and Optimization, University of Waterloo, Canada\\ 
\email{\{ijdehaan, kpashkovich\}@uwaterloo.ca}}

\begin{document}

\begin{titlepage}

\maketitle

\begin{abstract}

We study a class of Bayesian online selection problems with matroid constraints. Consider a vendor who has several items to sell, with the set of sold items being subject to some structural constraints, e.g., the set of sold items should be independent with respect to some matroid. Each item has an offer value drawn independently from a known distribution. Given distribution information for each item, the vendor wishes to maximize their expected revenue by carefully choosing which offers to accept as they arrive.

Such problems have been studied extensively when the vendor's revenue is compared with the offline optimum, referred to as the ``prophet''. In this setting, a tight 2-competitive algorithm is known when the vendor is limited to selling independent sets from a matroid \cite{KleinWein}. We turn our attention to the online optimum, or ``philosopher'', and ask how well the vendor can do with polynomial-time computation, compared to a vendor with unlimited computation but with the same limited distribution information about offers.

We show that when the underlying constraints are laminar and the arrival of buyers follows a natural ``left-to-right'' order, there is a Polynomial-Time Approximation Scheme for maximizing the vendor's revenue. We also show that such a result is impossible for the related case when the underlying constraints correspond to a graphic matroid. In particular, it is \ps-hard to approximate the philosopher's expected revenue to some fixed constant $\alpha < 1$; moreover, this cannot be alleviated by requirements on the arrival order in the case of graphic matroids.
\end{abstract}

\vspace{1cm}
\setcounter{tocdepth}{2} %
\tableofcontents

\end{titlepage}

\section{Introduction}
In this paper, we study the problem of Bayesian online selection subject to structural constraints given by matroids. Let us consider a scenario where a vendor posts and updates prices in order to maximize their profit subject to structural constraints. This type of problems is omnipresent in our everyday life. Consider vendors, e.g., big e-commerce platforms or  independent crafters, who sell items by posting prices. In order to maximize their profits, vendors usually calculate prices taking into account partial information about  potential buyers and the constraints on inventory, transportation networks, legal regulations, etc.

One of the most prominent examples of this setting is the \emph{single-item prophet inequality problem} \cite{Krengel}. In this problem, the vendor is selling one item, for which they observe a sequence of offers. The offers correspond to random variables $v_1, v_2, \ldots, v_n$ drawn independently from distributions known to the vendor. 
At each timestamp $t$, the vendor may choose to stop by selling their item and gaining the value $v_t$, or the vendor may choose to discard this offer and continue. The \emph{prophet} in this problem represents a person who knows the realizations of all offers $v_1, v_2, \ldots, v_n$ ahead of time. Moreover, the expected gain of the vendor is evaluated using the maximum gain of the prophet as a benchmark. 
In this scenario, the vendor can do at least half as well as the prophet and no better \cite{Krengel,SamuelCahn}.
This result generalizes to the setting of {\em matroid prophet inequalities}, in which the vendor is selling items from a matroid, and is limited to selling an independent set in this matroid~\cite{KleinWein}.

The classical prophet inequality problems make the vendor compete with the prophet, where only the prophet knows the realizations of $v_1, v_2, \ldots, v_n$ ahead of time. Clearly, such a competition between the vendor and the prophet is very unfair. Indeed, in many cases the advantage  of knowing all realizations cannot be alleviated through any efforts of the vendor. So let us change the benchmark and introduce the \emph{philosopher}. The philosopher does not know the realizations of $v_1, v_2, \ldots, v_n$ but has unlimited computational power. The central question for our work is as follows. \emph{How well can a vendor limited to polynomial time computation compete with the philosopher, where both know only the distributions of $v_1, v_2, \ldots, v_n$ but the latter has unlimited computational power?} Answering this question, we provide both positive and negative results for the ability of the vendor.

\subsection{Our Results}
In the single-item case, the vendor can achieve the same profit as the philosopher. Indeed, in this case the vendor can set a straightforward dynamic program that computes the optimal strategy.

Gupta posed the question of whether the vendor can achieve the same gain as the philosopher in the matroid setting, and more specifically in the graphic and laminar matroid settings~\cite{Gupta}. We answer this question in the negative for graphic matroids. We show that for graphic matroids, it is \ps-hard for the vendor to approximate the expected gain of the philosopher up to some fixed constant. Moreover, for the graphic matroids there is no arrival order which ``substantially increases competitiveness'' of the vendor.

On the positive side, we provide a Polynomial-Time Approximation Scheme (PTAS) for all laminar matroids with ``left-to-right'' arrival orders, in which elements from each constraint arrive consecutively. Furthermore, the provided PTAS also hold for arrival orders that are ``close'' to ``left-to-right'' orders, i.e., to orders where each element is contained only in constantly many bins on which the arrival order is not ``left-to-right.''  We note that our policy relies on the ``left-to-right'' order only in the analysis. We leave it as an  open question to determine whether there are substantially different requirements that guarantee our policy to lead to a PTAS.

    The defined order is called {\em left-to-right} because when the laminar matroid is drawn with all elements on one horizontal line, the elements can be arranged in order from left to right if and only if they are in a ``left-to-right'' orderƒing. 
These orderings are exactly those orderings that are obtained when the tree corresponding to the laminar family  is explored with depth-first-search.

 Let us provide an example of a Bayesian selection problem with a ``left-to-right'' arrival order. Consider a situation where items  correspond to clients, and the vendor knows clients' arrival order and distributions for their offers.
 Let the parameter $p$ represent some crucial resource and so determine restrictions on the number of clients the vendor can serve. In particular, let there be several critical thresholds $p_1$, $p_2$,\ldots  for $p$ and limits $\gamma_1$, $\gamma_2$,\ldots. For each $i$, if the value of $p$ drops below $p_i$ at some timestamp then we can serve at most $\gamma_i$ clients between this timestamp and the next timestamp when the value of $p$ is again at least $p_i$. In Figure~\ref{fig:level_sets}, one can see the example of how the parameter $p$ changes over time and the corresponding laminar family. We note that the arrival order in the ``production constrained Bayesian selection'' from \cite{AnariLMBS} corresponds to the case when $p$ is a resource that is being delivered to us over time, see Figure~\ref{fig:level_sets_product}, plus one additional global constraint. 

\begin{figure}[h]
\centering
\includegraphics[width=0.65\textwidth]{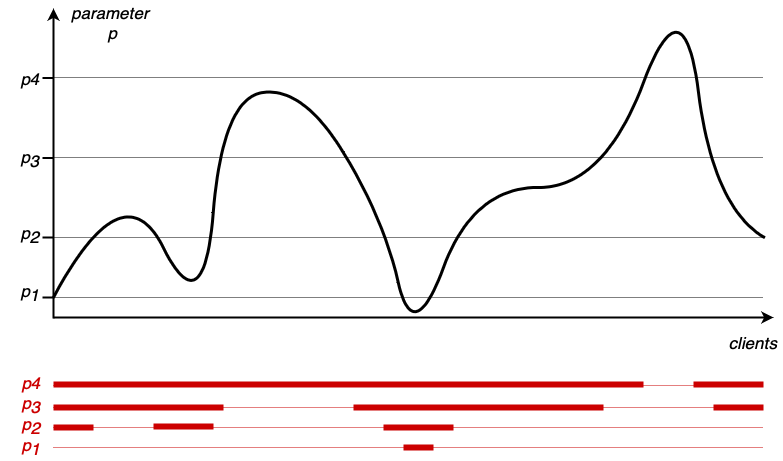}
\caption{Here, the horizontal axis is associated with clients and so with the timestamps, since their arrivals provide a measure for time. The vertical axis is associated with the value of the parameter $p$. The critical thresholds $p_1$, $p_2$, $p_3$ and $p_4$ for the value of $p$ are depicted on the vertical axis. Below the picture of the graph, one can find the illustration for the corresponding laminar matroid. In particular, the picture below contains an interval for each critical threshold and the timestamps, when the value of $p$ drops below the threshold and the next timestamp when the value of $p$ reaches the value of the threshold.}
\label{fig:level_sets}
\end{figure}

\begin{figure}[h]
\centering
\includegraphics[width=0.65\textwidth]{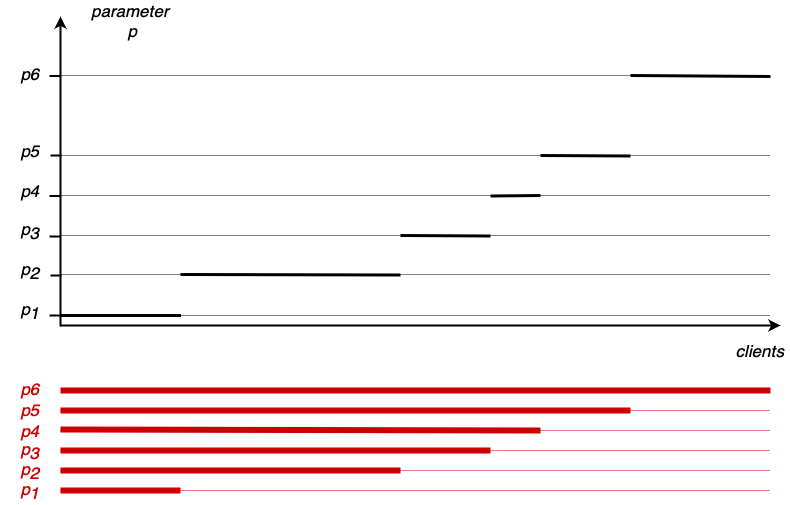}
\caption{Here the legend of the figure is  the same as in Figure~\ref{fig:level_sets}. The structure of the function for the parameter $p$ corresponds to the situation when $p$ represents a resource that is being delivered over time, and is spent only on servicing clients. In this scenario, the vendor can serve only a certain number of clients until the next delivery of the resource.}
\label{fig:level_sets_product}
\end{figure}

Our PTAS holds for more general matroid Bayesian selection problems  than the ones described above. We can handle the cases where the limits $\gamma_1$, $\gamma_2$,\ldots are functions of timestamps when the value of $p$ drops below $p_i$. Moreover, we allow the arrival order to be more complicated than just ``left-to-right''. One can imagine that the vendor runs several services in parallel, which mostly depend on different resources, however each type of service depends on at most a constant number of resources that are needed also for other types of service. In this situation, as long as the dependence on the ``overarching'' resources leads to a laminar matroid, our policy provides a PTAS for the arising Bayesian selection problem.

Thus our results provide a generalization and extension of the previous results for laminar matroids. The study of laminar Bayesian selection problems  was initiated by Anari et al. \cite{AnariLMBS}. They gave a PTAS for the special cases of bounded-depth laminar matroids and production constrained selection.
Note that both of these special cases have required structure on the laminar family, with production constrained selection also having a special arrival order. In contrast, our results do not impose any requirements on the structure of the laminar matroids, but only on the arrival order.

\subsection{Our Techniques}

\subsubsection{PTAS for Left-to-Right Laminar Bayesian Selection} 
We capture the optimal online policy for left-to-right laminar Bayesian selection with an exponential-sized linear program introduced in \cite{AnariLMBS}.
While this linear program is too large to solve efficiently, we create a polynomial-sized relaxation by partitioning bins into ``big'' and ``small'' based on their capacities.
For each maximal ``small'' bin, we require the constraints from the exponential-sized linear program to hold exactly.
But for each ``big'' bin, we only require its capacity constraints to ``hold in  expectation.''

We show that by solving the linear program, we can efficiently obtain an online policy that is feasible for all ``small'' bins and has total expected gain equal to the optimal value of the linear program. In fact, the obtained online policies are optimal for the maximal ``small'' bins subject to changes in the value distributions. The argumentation about the above online policies goes generally along the lines of the analysis in \cite{AnariLMBS}.
One of our main technical contributions is showing that, with some pre-processing on capacities, ``big'' bins are unlikely to be violated by the obtained policies, see Lemma~\ref{lem:failure-prob}. To show this, we need concentration bounds on the number of items selected from each maximal ``small'' bin.  The fact that the obtained policies act optimally on the maximal ``small'' bins does not guarantee us that selections of items are negatively correlated, see Figure~\ref{fig:positive-correlation-uniform}, and not even that the selection of an item is negatively correlated with the number of items selected before it, see Figure~\ref{fig:positive-correlation-laminar}. Thus, we cannot rely on the analysis from~\cite{AnariLMBS} and have to develop new tools. Nevertheless, we are able to show a more ``global'' version of negative dependence; we show that the number of items selected from each small bin is concentrated closely around the mean. Roughly speaking, we obtain this by showing that the selection of an item has a limited impact on the expected number of the items selected after it, see Lemma~\ref{lem:firstuseless}.  This gives us the building blocks needed for Chernoff type results, see Lemma~\ref{lem:concentration}, which we use to bound the probability of ``big'' bins being violated in Lemma~\ref{lem:failure-prob}.

These concentration results are only possible for restricted arrival orders. We show in Theorem \ref{thm:badconcentration} that there exist laminar Bayesian online selection instances where the number of elements selected by the optimal policy is anti-concentrated. Due to the  constructions done in the proof of Theorem \ref{thm:badconcentration}, straightforward counting arguments show that given natural numbers $r$ and $n \in \Omega(r^2)$, for a randomly chosen  laminar matroid over $n$ elements with rank $r$ and a randomly chosen order, asymptotically almost surely one can choose value distributions and capacities such that the resulting instance exhibits anti-concentration for the number of elements selected by the optimal policy.
This is a significant roadblock on the way to giving a good approximation for laminar Bayesian online selection problems with no restriction on arrival order.
Most known results for problems of this type first break the problem into smaller pieces, solve the philosopher's problem optimally on each part, and then use some concentration result to show that combining the solution for these parts is unlikely to produce infeasible solutions for the global problem.
Without strong concentration on the number of elements selected by the philosopher, this framework cannot work - further new ideas are needed to tackle the general problem.

\begin{figure}[h]
\centering
\includegraphics[width=0.61\textwidth]{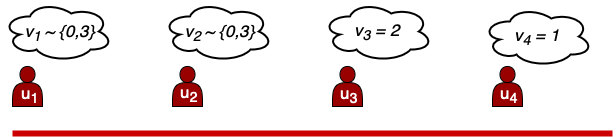}
\caption{Here, the vendor receives four clients $u_1$, $u_2$, $u_3$ and $u_4$ in the corresponding order, but can serve at most $2$ clients. The distributions for the offers are as in the figure, e.g., $u_2$ offers $3$ with probability $0.5$ and otherwise  offers $0$, $u_4$ always offers $1$. Let $X_i$, $i=1,\ldots,4$ be the event (and the corresponding indicator variable) that $u_i$ was served by the optimal online policy. We have $\prob[X_3]=3/4$, $\prob[X_4]=1/4$ while $\prob[X_3\land X_4]=1/4$. Thus, we have $\covariance(X_3,X_4)=1/16$. }
\label{fig:positive-correlation-uniform}
\end{figure}

\begin{figure}[h]
\centering
\includegraphics[width=.9\textwidth]{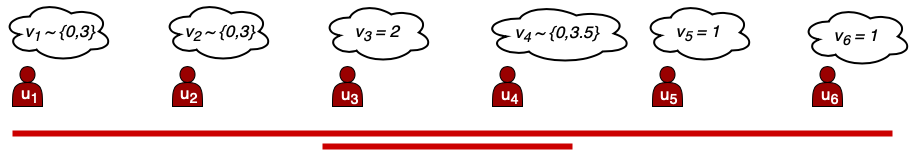}
\caption{Here, the vendor receives six clients $u_1$,\ldots, $u_6$ in the corresponding order, but can serve at most $3$ clients. Moreover, among $u_3$ and $u_4$ the vendor can serve at most one client. The distributions for the offers are as on the figure, e.g. $u_2$ makes offer $3$ with probability $0.5$ and otherwise makes  offer $0$, $u_5$ always offers $1$. Let $X_i$, $i=1,\ldots,6$ be the event (and the corresponding indicator variable) that $u_i$ was served by the optimal online policy. We have $\prob[X_1]=\prob[X_2]=1/2$, $\prob[X_3]=1/4$, $\prob[X_4]=3/8$, $\prob[X_1\land X_4]=\prob[X_2\land X_4]=1/4$ and $\prob[X_3\land X_4]=0$. Thus, we have $\covariance(X_1+X_2+X_3,X_4)=1/32$. }
\label{fig:positive-correlation-laminar}
\end{figure}

\subsubsection{Hardness of Graphic Matroid Bayesian Selection} 
To obtain \ps-hardness result for approximating Bayesian online selection for graphic matroids, we reduce from Stochastic MAX2SAT (\stwosat). The arrival order of edges is organized in three phases. In the first phase, the optimal policy for the constructed instance would need to make decisions that are equivalent to assigning the variables in \stwosat \texttt{True} and \texttt{False} values. In the second phase and the third phase, the optimal policy would make greedy decisions, with expected payoff depending on the choices made in the first phase. Having shown that the behaviour of the optimal policy indeed follows the above rules, the expected gain of any such policy would be equal to the expected number of satisfied clauses plus a fixed term and a negligible error term. Our analysis to estimate the expected gain follows the lines of the work~\cite{Papadimitriou2021} on stochastic online matching. Later, we show that for graphic matroids, there are no orders like left-to-right orderings for laminar matroids. In particular, without knowing value distributions  we cannot associate an arrival order to each graphic matroid such that the resulting class of instances admits a PTAS. To obtain this result, we show that we can embed our hardness instance into a sufficiently large complete graph for any arrival order on that complete graph. 

\subsection{Further Related Work}

Study of prophet inequalities was initiated by \cite{Krengel} in which they considered adaptive algorithms for the single-item setting. 
Several years later, a simple single-threshold $\frac{1}{2}$-competitive algorithm was given in \cite{SamuelCahn}.
In the decades following, many variants of the classic prophet inequality problem have been considered.
We give a brief snapshot of work on variants related to this paper, and recommend reading surveys on the area for a more complete picture \cite{LucierSurvey,HillKertzSurvey,VariantSurvey}.

In recent years, there has been much work dedicated to the study of prophet inequalities under combinatorial constraints. Beginning with uniform matroid constraints \cite{HajiaghayiUniformPI,SaeedUniformPI}, there has been exploration of general matroid constraints \cite{KleinWein}, knapsack constraints \cite{feldman2016online}, matching constraints \cite{alaei2012online,gravin2019prophet,ezra2020online}, downwards-closed constraints \cite{rubinstein2017combinatorial,rubinstein2016beyond}, and many more.
Many other variations of the classic prophet inequality problem have been studied extensively. Some examples include settings where the arrival order of elements is random \cite{ehsani2018prophet,esfandiari2017prophet,correa2021prophet,azar2018prophet}, where every element has its value drawn independently from the same distribution, \cite{abolhassani2017beating,correa2017posted,HillKertzPI}, single sample prophet inequalities \cite{rubinstein2019optimal,azar2014prophet}, and non-adaptive prophet inequalities \cite{chawla2020non,pashkovich2023non}.

Our work is far from the first to consider approximation of the online optimum for Bayesian online selection problems.  Online stochastic weighted bipartite matching has been extensively studied in the last several years, with a series of improved approximation factors given in \cite{Papadimitriou2021,LovettOnlineMatching,WajcMatchingRounding}. In a variant of the single-item prophet inequality problem where the order of elements is unknown and uniformly random, there is a Polynomial-Time Approximation Scheme \cite{ProphetSecretary}. \cite{OrderSelectionHard} shows that even in the single-item prophet inequality problem, it is NP-hard to select the best order to observe elements in. Additional approximations on the online optimum in stochastic and Bayesian online selection problems have been given in a variety of settings \cite{feng2022near,segev2021efficient,feng2021two,fu2018ptas}. 

During the preparation of our paper, we became aware also of an unpublished \ps-hardness result for graphic Bayesian online selection obtained independently by another group~\cite{Wajc2023}.

\section{Problem Definition and Preliminaries}

In this paper, we consider structural constraints defined by matroids. Each matroid $\mat$ is defined by a ground set $U$ and a collection of subsets of $U$, which are called \emph{independent sets}.

In this paper, we work with two types of matroids: laminar and graphic.    We say that a matroid $\mat$ over a ground set $U$ is \emph{laminar} if there is some laminar family of sets $\lam$ over $U$ and a capacity function $c : \lam \rightarrow \mathbb{N}$ such that a set $I$ is independent if and only if $|I \cap A| \leq c(A)$ for all $A \in \lam$. 
    In this case, we write $\mat = (U, \lam, c)$. Given a laminar matroid $\mat = (U, \lam, c)$, we call the sets in $\lam$ \emph{bins}. The \emph{depth} of a bin $A\in \lam$ is defined as $|\{B\,:\, A\subseteq B,\, B\in \lam\}|$.
We say that a matroid $\mat$ over a ground set $U:=E$ is \emph{graphic} if there is some graph $G=(V,E)$ such that a set $I\subseteq E$ is independent if and only if $(V,I)$ has no cycles. For further reading about matroids and their properties, we refer to~\cite{oxley}.

From now on, we refer to the vendor as a \emph{gambler} and to the items as \emph{elements}. In the matroid Bayesian online selection problem, the gambler is given the matroid $\mat$ over a ground set $U=\{u_1,\ldots, u_n\}$, an arrival order of the elements in $U$, and for each element $u_i\in U$, a  distribution $F_i$ for its value $v_i$.
Over the course of the game, the gambler  maintains an independent set $I$ which starts as the empty set. The set $I$ consists of all elements selected by the gambler at the given timestamp.
    
    Elements arrive one by one in the given order. 
    When $u_i$ arrives, the value $v_i \sim F_i$ is drawn from its distribution independently from values of other elements and presented to the gambler. 
    The gambler may then choose to either select the element as long as independence is maintained, updating $I \leftarrow I \cup \{u_i\}$ and gaining the associated value $v_i$, or reject the element, gaining no value. 
    This decision is final and may not be changed later in the game.

    The goal of the gambler is to gain as much value as possible over the course of the game. Let $OPT_{\mat}$ indicate the maximum expected gain  that the gambler can achieve for the instance $\mat$. Here, we abuse notation and associate the matroid Bayesian selection instance to the underlying matroid $\mat$ when the value distributions and arrival order are not relevant or are clear from the context. 

    The \textsc{Laminar Matroid Bayesian Selection} (\lmbs) problem is the matroid Bayesian online selection problem restricted to laminar matroids.
 The \textsc{Graphic Matroid Bayesian Selection} (\gmbs) problem is the matroid Bayesian online selection problem restricted to graphic matroids.

\section{Left-to-Right Laminar Bayesian Selection}
\label{sec:PTAS}

In this section, we show that the laminar Bayesian online selection problem with certain types of arrival orders admits a PTAS.

\begin{definition}
    Given a laminar matroid $\mat = (U, \lam, c)$ and an ordering $u_1, u_2, \ldots, u_n$ of the elements $U$, we say that this is a \emph{left-to-right ordering} if the following condition holds:
    \begin{quote}
            For every bin $A \in \lam$ of the laminar family, if $u_i, u_j \in A$ for $i < j$, then for every $k \in [i, j]$, we must have $u_k \in A$. 
    \end{quote}
\end{definition}

    In other words, the left-to-right ordering captures the rule that once elements from some bin start to arrive, they must not stop until they have all arrived.
We say that an \lmbs \emph{instance has a left-to-right arrival} order (or equivalently it is a \emph{left-to-right \lmbs instance}) if the elements arrive according to a  left-to-right ordering. 

In this section, we assume  that each distribution in the input is atomic. Moreover, each distribution is given to us explicitly as the list of values and the corresponding probabilities. This allows us to efficiently represent all distributions and to perform computations on them in polynomial time.

\subsection{Warm-Up: QPTAS}

To design a QPTAS and PTAS, we rely on the concept of states. A \emph{state} has an entry for each considered bin $A$ in $\lam$, and this entry equals $c(A)$ minus the number of currently selected elements from $A$. Informally,  a state represents the remaining capacities of the bins under consideration.

In an \lmbs instance, given a timestamp let us call a bin $A$  {\em active} if at least one element of $A$ already arrived and there are still some elements of $A$ left to arrive. The next lemma is based on the fact that in an \lmbs instance, we only need to keep track of states for active bins.

\begin{lemma} \label{lem:depthdp}
    Let  $\mat = (U, \lam, c)$ be a left-to-right \lmbs instance  with depth at most $L$. 
    Then the optimal gain and an optimal policy of the gambler in $\mat$ can be computed in $n^{\mathcal{O}(L)}$ time. 
\end{lemma}
\begin{proof}
    The computation can be done with a dynamic programming algorithm. At each timestamp, we keep track of states for currently active bins. The left-to-right ordering guarantees that at each timestamp we have at most $L$ active bins. Since $c(A)$ is at most $n$ for all $A \in \lam$, we have that the total number of possible states is at most $n^{\mathcal{O}(L)}$.
\end{proof}

\begin{corollary}
    The optimal gain and an optimal policy of the gambler in a constant-depth left-to-right LMBS instance can be computed in polynomial time.
\end{corollary} 

\begin{lemma} \label{lem:logdepth}

Let  $\mat = (U, \lam, c)$ be an \lmbs instance and $\alpha$ be  in $ (0, 1)$. We can, in polynomial time, construct an \lmbs instance $\mat' = (U, \lam', c')$ such that:
\begin{itemize}
    \item $\lam' \subseteq \lam$  
    \item for all $A$, $B \in \lam'$ with $A \subsetneq B$, we have  \begin{equation}\label{eq:separation_property}\tag{SEP}
        c'(A) \leq \lceil \alpha \cdot c'(B) \rceil
    \end{equation}
    \item all independent sets in $\mat'$ are also independent in $\mat$
    \item $\alpha \cdot OPT_{\mat} \leq OPT_{\mat'} \leq OPT_{\mat}$
    \item if $\mat$ is a left-to-right instance then $\mat'$ is also a left-to-right instance.
\end{itemize}
\end{lemma}
\begin{proof}[Proof of Lemma \ref{lem:logdepth}]
For the construction, let us start with empty $\lam'$ and process the  bins in $\lam$ in a non-increasing order of their cardinality.

Let $A\in\lam$ be the currently processed bin. 
If $A$ is not contained in any other bins from $\lam'$, add $A$ to $\lam'$ and set $c'(A) := c(A)$.
Otherwise, let $B \in \lam'$ be the inclusion-wise minimal bin in $\lam'$ containing $A$. 
If $c(A) \geq c'(B)$, we move to the next unprocessed bin from $\lam$. 
Otherwise,  we add $A$ to $\lam'$ and set $c'(A) := \min\{c(A), \lceil \alpha \cdot c'(B) \rceil\}$. 
This transformation yields the desired separation property~\eqref{eq:separation_property}.

Note that any independent set in $\mat'$ is also independent in $\mat$. Indeed,  if $A$ is in $\lam'$ then we have $c'(A)\leq c(A)$. If $A$ is in $\lam\setminus\lam'$ then $c(A)\leq c'(B)$ where $B\in \lam'$ is the inclusion-wise minimal bin in $\lam'$ containing $A$.
So, we have $OPT_{\mat'} \leq OPT_{\mat}$.
Finally, to get $OPT_{\mat'} \geq \alpha \cdot OPT_{\mat}$, note that for every bin $A\in \lam'$, we have that $c'(A)\geq \alpha \cdot c(A)$.
\end{proof}

\begin{corollary} \label{corr:qptas}
    There is a QPTAS for left-to-right \lmbs.
\end{corollary}
\begin{proof}
    Given a left-to-right \lmbs instance $\mat$, we can use Lemma~\ref{lem:logdepth} with $\alpha=1-\epsilon$ to get an instance $\mat'$. 
    By~\eqref{eq:separation_property} with $\alpha=1-\epsilon$, the depth of $\mat'$ is $\mathcal{O}(\log_{1/(1 - \epsilon)}(n))$. Now we can use Lemma~\ref{lem:depthdp} to compute the optimal gain and policy for $\mat'$ in $n^{\mathcal{O}(\log_{1/(1 - \epsilon)}(n))}$ time. By Lemma~\ref{lem:logdepth}, we obtain a $(1 - \epsilon)$-approximation to the optimal gambler's gain for $\mat$.
\end{proof}

\subsection{Linear Programming Formulation and Rounding Algorithm}

A naive dynamic programming algorithm for \lmbs can be written as in the proof of Lemma~\ref{lem:depthdp}. However, the number of considered states makes this approach  not viable. To circumvent this, we build on  a linear programming formulation from~\cite{AnariLMBS}. The linear program from~\cite{AnariLMBS} encodes dynamic programming ideas as in Lemma~\ref{lem:depthdp}. Their linear program allows us to construct a relaxation by  decomposing a given instance of $\lmbs$ into tractable parts and concentrating on optimal policies for each of these tractable parts. For the sake of completeness, we present the linear program and relevant results from~\cite{AnariLMBS} in this section.

\subsubsection{Small and Big Bins}\label{sec:smallbigbins}
To decompose an instance of $\lmbs$ into tractable parts, we partition $\lam$ into ``big'' and ``small'' bins based on their capacities.
If a bin is small, we use a linear program based on dynamic programming to guarantee its ``feasibility.''
If a bin is big, we enforce that its ``feasibility'' is guaranteed in  expectation.

Consider a threshold $K\in\mathbb{N}$.
Call a bin $A \in \lam$ \emph{big} if $c(A) \geq K$, and \emph{small} otherwise.
Let $\largebin \subseteq \lam$ indicate the set of big bins and $\smallbin \subseteq \lam$ indicate the set of inclusion-wise maximal small bins.

Without loss of generality, we assume that each element of $U$ is contained in some small bin. If we have an element not contained in any small bin, then we can add a bin containing only this element. So, we can assume that $\smallbin$ partitions $U$.

\subsubsection{Feasibility on Small Bins}
Let us consider a maximal small bin $B\in \smallbin$. Let $\lam^B = \{B' \in \lam \,:\, B' \subseteq B\}$ be the set of bins contained in $B$.
To capture states with respect to $B$ at each timestamp, we use $\sets^B \subseteq \mathbb{Z}^{\lam^B}$ to indicate the set of feasible states. Again, each vector  $\vecs \in \sets^B$ has an entry for every $A\in \lam^B$ which represents $c(A)$ minus the number of selected elements from $A$. Note that $|\sets^B|$ is at most $n^{\mathcal O(c(B))}$ because no more than $c(B)$ elements can be selected from $B$ by an online policy.

Let $\vec{d}_t \in \{0, 1\}^{\lam^B}$ be the indicator vector for the bins in $\lam^B$ containing the element $u_t$.
Let us define the set of \emph{forbidden neighboring states}
\[\partial \sets^B:=\{\vec{f}  \in \mathbb{Z}^{\lam^B}\setminus \sets^B\,:\,\text{there are }u_t \in B \text{ and } \vecs \in  \sets^B \text{ such that } \vecs = \vec{f} + \vec{d}_t\}\,.\]
In other words, $\partial \sets^B$ is the set of all states that can be reached by an online policy immediately after violating feasibility.

We use {\em allocation variables} $\varx_t(\vecs, v)$, which represent the probability that the gambler accepts the $t$th item $u_t$ and the state upon its arrival is $\vecs$, conditioned on $v$ being the realized value of $v_t$.
We then use $\varx_t(v)$ to represent the conditional probability that we accept $u_t$, conditioned on $v$ being the realized value of $v_t$.
The {\em state variables} $\vy_t(\vecs)$ represent the probability that the gambler is at state $\vecs$ upon the arrival of $u_t$.

Now we introduce a polytope $\mathcal{P}^B$ to capture transitions between states through an application of an online policy on the bin $B$. The polytope $\mathcal{P}^B$ is defined by the following constraints:
\begin{align*}
&\varx_t(v) = \sum_{\vecs \in \sets^B} \varx_t(\vecs, v) & \quad & \forall v, u_t \in B \\
&0 \leq \varx_t(\vecs, v) \leq \vy_t(\vecs) && \forall v, \vecs \in \sets^B, u_t \in B \\
&\vy_{t+1}(\vecs) = \vy_t(\vecs) - \expect_{v_t}[\varx_t(\vecs, v_t)] + \expect_{v_t}[\varx_t(\vecs + \vec{d}_t, v_t)] && \forall \vecs \in \sets^B, u_{t}, u_{t+1} \in B \\
&\vy_{t_0}([c(A)]_{A \in \lam^B}) = 1 && t_0 = \min\{t : u_t \in B\} \\
&\vy_t(\vecs) = 0 && \forall \vecs \in \partial \sets^B, u_t \in B\,.
\end{align*}
One can see that any  online policy for the restriction of $\mat$ on $B$ induces a feasible point in $\mathcal{P}^B$. Let us show that the opposite also holds.

\begin{proposition}[Proposition 2.1 from \cite{AnariLMBS}] \label{prop:solvelp}
    Given a bin $B\in \smallbin$ and a point $\{ \varx_t(\vecs, v), \varx_t(v), \vy_t(\vecs)\} \in \mathcal{P}^B$, there is an online policy for the restriction of $\mat$ on $B$ which guarantees the expected gambler's gain to be $\sum_{u_t\in B} \expect_{v_t} [v_t \cdot \varx_t(v_t)]$ and the expected number of selected elements to be $\sum_{u_t \in B} \expect_{v_t}[\varx_t(v_t)]$. 
    Moreover, this policy can be found in polynomial time.
\end{proposition}

\begin{proof}
Let us outline the desired policy. Let us assume that the gambler is at the state $\vecs$ at the timestamp $t$ and $v$ is the realized value of $v_t$. If $\vy_t(\vecs)$ is nonzero, the gambler accepts the element $u_t$ with probability ${\varx_t(\vecs, v)}/{\vy_t(\vecs)}$. Otherwise, the gambler rejects the element $u_t$. It is straightforward to see that this policy  is feasible and guarantees the required gain and the expected number of selected elements.
\end{proof}

\subsubsection{Feasibility on Big Bins}
For each maximal small bin $B \in \smallbin$, we introduce a variable $\mathcal{N}_B$ to represent an upper bound on the expected number of elements selected from $B$. Using these variables we construct our final linear program.
\begin{align*}
&\text{maximize} & \quad & \sum_{t=1}^n \expect_{v_t} [v_t \cdot \varx_t(v_t)] \tag{$LP$} \label{lp2} \\
&\text{subject to} && \sum_{A \in \smallbin : A \subseteq B} \mathcal{N}_{A} \leq c(B) & \forall B \in \largebin \\
&&& \sum_{u_t \in B} \expect_{v_t}[\varx_t(v_t)] \leq \mathcal{N}_B & \forall B \in \smallbin \\
&&&\{ \varx_t(\vecs, v), \varx_t(v), \vy_t(\vecs)\} \in \mathcal{P}^B & \forall B \in \smallbin 
\end{align*}
Note that in~\eqref{lp2} the first type of constraints are global ex-ante constraints. They enforce the ``feasibility'' of big bins in expectation. The second type of constraints enforces that for each $B\in \smallbin$ the variable $\mathcal{N}_B$ is a correct upper bound on the expected number of selected elements. Finally, the third type of constraints guarantees ``feasibility'' on all small bins. One can see that any online policy induces a feasible solution for~\eqref{lp2}.

\subsubsection{Algorithm}
Now we are ready to state the algorithm. We pre-processes an input \lmbs instance as in Lemma~\ref{lem:logdepth} to guarantee the separation property~\eqref{eq:separation_property} with $\alpha=(1-\epsilon)$. After that we divide the bins into small and big as in Section~\ref{sec:smallbigbins}. We decrease the capacities for big bins in order to introduce ``slack'' for our online policy to satisfy their constraints. The capacities of small bins are satisfied by our online policy due to the constraints in the linear program \eqref{lp2}.

\begin{algorithm}[H]
\caption{ALG($\mat=(U,\lam,c)$, $K$, $\epsilon$)}\label{alg:PTAS}
Obtain $\mat'=(U,\lam',c')$ from $\mat=(U,\lam,c)$ as in  Lemma~\ref{lem:logdepth} based on $\alpha=(1 - \epsilon)$ \label{alg:PTAS-step-1}\;
Compute maximal small bins $\smallbin$ and big bins $\largebin$ for $\mat'$ as in Section~\ref{sec:smallbigbins} based on $K$ \label{alg:PTAS-step-2}\;
Obtain $\mat''=(U,\lam',c'')$ from $\mat'=(U,\lam',c')$ by setting $c''(B):=(1-\epsilon)\cdot c'(B)$ for $B\in \largebin$; and $c''(B):=c'(B)$ for $B\in \lam'\setminus\largebin$ \label{alg:PTAS-step-3}\;
Compute the optimal solution $\{\varx_t^{\star}(\vecs, v_t), \varx_t^{\star}(v_t), \vy^*_{t}(\vecs),  \mathcal{N}^*_A\}$ for the linear program~\eqref{lp2} with respect to $\mat''$ \label{alg:PTAS-step-4}\;
Extract an online policy from $\{\varx_t^{\star}(\vecs, v_t), \varx_t^{\star}(v_t), \vy^*_{t}(\vecs)\}$ for each $B\in \smallbin$ as in Proposition~\ref{prop:solvelp} \label{alg:PTAS-step-5}\;
Independently run the obtained online policies on $B\in \smallbin$. Select an element as long as the obtained policy suggests to select it and the selection is feasible with respect to $\mat'$ \label{alg:PTAS-step-6}.
\end{algorithm}

\subsection{Analysis}

In this section, we estimate the approximation guarantees achieved by Algorithm~\ref{alg:PTAS} with respect to an optimal online policy.

From now on, we assume for simplicity that each distribution is non-atomic.
This simplifies the analysis of the optimal online policy, because for non-atomic distributions the optimal online policy is both unique and deterministic.
A similar analysis applies to atomic distributions - the randomness in optimal online policies can be accounted for by considering small dispersals of the values in each distribution.
As the dispersals grow smaller, the behavior of the optimal online policy for the non-atomic case becomes the behavior of the optimal online policy for the atomic case.

\subsubsection{Losses from Discarding Elements in Step~\ref{alg:PTAS-step-6} of Algorithm~\ref{alg:PTAS}}

The key challenge of analyzing Algorithm~\ref{alg:PTAS} is to estimate the expected value of elements that were not selected in Step~\ref{alg:PTAS-step-6} due to the feasibility restrictions of $\mat'$. Lemma~\ref{lem:failure-prob} shows that each element has a rather small probability of being not selected  due to the feasibility restrictions of $\mat'$. To show this, we first need the following technical result about the concentration on the number of elements chosen in each small bin. We defer its proof to Section~\ref{sec:concentration}.

\begin{lemma} \label{lem:concentration}
    Given a maximal small bin $B \in \smallbin$ and an element $u_t\in B$, let $X_t$ be the event that the online policies computed in Step~\ref{alg:PTAS-step-5} of Algorithm~\ref{alg:PTAS} suggest to select $u_t$.
    Then, for every $\alpha>0$ we have that
    \[
        \expect[e^{\alpha \cdot X_B}] \leq e^{(e^\alpha - 1) \cdot \expect[X_B]}\,,
    \]
    where $X_B:=\sum_{u_t\in B}X_t$.
\end{lemma}

We note that Lemma \ref{lem:concentration} crucially relies on the ``left-to-right'' arrival order assumption.
If this assumption is dropped, instances with essentially worst-possible concentration can be constructed.
\begin{theorem} \label{thm:badconcentration}
    For any $r \geq 1$ and $\epsilon > 0$, there exists an instance of laminar Bayesian online selection $\mat = (U, \lam, c)$ with rank$(\mat) = r$ and $|U| \leq (r+1)^2$ such that
    \[
        \prob\left[\,|OPT(\mat)\right| = 0\,],\quad \prob\left[\,|OPT(\mat)| = r\,\right]\, \in\, [\frac{1}{2} - \epsilon, \frac{1}{2} + \epsilon],\,
    \]
    where $OPT(\mat) \subseteq U$ is a random variable denoting the items selected by an optimal algorithm.
\end{theorem}

We provide the proof of Theorem~\ref{thm:badconcentration} in appendix.

\begin{lemma}\label{lem:failure-prob} In left-to-right \lmbs, for each $u_t\in U$, the probability that upon the arrival of $u_t$ the online policies computed in Step~\ref{alg:PTAS-step-5} suggest to select $u_t$ but it is not feasible with respect to $\mat'$ is at most $3/(K\cdot \epsilon^3)$.
\end{lemma}
\begin{proof}

    For each item $u_i\in U$, let $X_i$ be the event that the online policies computed in Step~\ref{alg:PTAS-step-5} suggest to select $u_i$. Given a bin $B\in \lam'$, let $X_B := \sum_{u_i \in B} X_i$.

    By the design of Algorithm~\ref{alg:PTAS}, for every small bin $B\in \smallbin$, the sum $\sum_{u_i\in B} X_i$ is always at most $c'(B)$. 
Now, consider any big bin $B\in \largebin$ with $u_t \in B$. Let us  bound the probability that $X_B$ exceeds $c'(B)$. First, we note that $\expect[X_B] \leq (1 - \epsilon) \cdot c'(B)$ due to the structure of~\eqref{lp2} in Step~\ref{alg:PTAS-step-4} and $c''(B)\leq (1 - \epsilon) \cdot c'(B)$.  
    By Lemma~\ref{lem:concentration}, we have strong bounds on the concentration of $X_A$ for every maximal small bin $A \in \smallbin$ with $A \subseteq B$. For every $\alpha > 0$, we have $
        \expect[e^{\alpha X_A}] \leq e^{(e^\alpha - 1) \cdot \expect[X_A] }$.
    
    Note that the suggestions of online policies from Step~\ref{alg:PTAS-step-5} on different maximal small bins $A\in\smallbin$ are independent.
    Since we assumed that every element is contained in some maximal small bin we have that
    \[
        \expect[e^{\alpha\cdot X_B}]=
        \expect[e^{\alpha\cdot \sum_{A\in \smallbin,\,A\subseteq B}X_A}] =\prod_{A\in \smallbin,\,A\subseteq B} \expect[e^{\alpha \cdot X_A}]\leq e^{(e^\alpha - 1) \cdot \sum_{A\in \smallbin,\,A\subseteq B}\expect[X_A]} 
        \leq e^{(e^\alpha - 1) \cdot (1-\epsilon)\cdot c'(B)}.
    \]
  Hence, by Chernoff type estimation we have
    \begin{align*}
        \prob[X_B > c'(B)] \leq e^{(e^\alpha - 1) \cdot (1-\epsilon)\cdot c'(B)}/{e^{\alpha\cdot c'(B)}}=e^{\left((e^\alpha -1)(1-\epsilon)-\alpha\right)\cdot {c'(B)}}
    \end{align*}
    and by selecting $\alpha=\ln(1+\epsilon)$ and by $\ln(1+\epsilon)\geq \epsilon/(1+\epsilon/2)$
      \begin{equation*}
        \prob[X_B > c'(B)] \leq e^{\left(\epsilon\cdot (1-\epsilon)-\ln(1+\epsilon)\right)\cdot {c'(B)}}\leq e^{-\frac{\epsilon^2}{(2+\epsilon)}\cdot {c'(B)}}\,.
      \end{equation*}

    Recall that for any two big bins $B$, $B'\in \largebin$ such that $B \subsetneq B'$ we have $c'(B)\leq (1-\epsilon)\cdot c'(B')$. Recall also that for any big bin $B\in\largebin$ we have $c'(B) \geq K$. Thus, using our previous estimates we have
    \[\sum_{\substack{B\in \largebin\\u_t\in B}}\prob[X_B > c'(B)]\leq \sum_{j=1}^{\infty}e^{-\frac{\epsilon^2}{(2+\epsilon)}\cdot {K/(1-\epsilon)^j}}\leq \sum_{j=1}^{\infty}\frac{2+\epsilon}{K \cdot\epsilon^2}(1-\epsilon)^j\leq \frac{3}{K\cdot \epsilon^3}\,,\]
    finishing the proof.
\end{proof}

\subsubsection{Losses from Transformations in Steps~\ref{alg:PTAS-step-1}-\ref{alg:PTAS-step-3} of Algorithm~\ref{alg:PTAS}}

There are two steps in Algorithm~\ref{alg:PTAS} where we transform the original \lmbs problem: Step~\ref{alg:PTAS-step-1} and  Step~\ref{alg:PTAS-step-3}.  The next lemma captures how the optimal gain changes between these transformations. We are also interested in the gain estimations obtained through~\eqref{lp2}.
So, let us define~$LP_{\mat'}$ and~$LP_{\mat''}$ to indicate the optimal objective values of~\eqref{lp2} with respect to $\mat'$ and $\mat''$, respectively.

\begin{lemma}\label{lem:preprocessed-matroid-gain}
    In Algorithm~\ref{alg:PTAS}, the following inequality holds
\[
    LP_{\mat''} \geq (1-\epsilon)^2 \cdot OPT_{\mat}\,.
\]
\end{lemma}
\begin{proof}
    By Lemma~\ref{lem:logdepth}, we have 
\[
    OPT_{\mat'} \geq (1 - \epsilon) \cdot OPT_{\mat}.
\]
By definition of~\eqref{lp2} we have $LP_{\mat'} \geq OPT_{\mat'}$.
To finish the proof, note that any feasible point in the polytope~\eqref{lp2} for $\mat'$ becomes feasible in~\eqref{lp2} for $\mat''$ through a  multiplication by a factor of $(1 - \epsilon)$.
Thus, we obtain $LP_{\mat''} \geq (1 - \epsilon) \cdot LP_{\mat'}$.
\end{proof}

\subsubsection{Characterizations of Policies in Step~\ref{alg:PTAS-step-5} of Algorithm~\ref{alg:PTAS}} To estimate the performance of Step~\ref{alg:PTAS-step-6} in Algorithm~\ref{alg:PTAS}, we need first to understand the nature of the online policies obtained in Step~\ref{alg:PTAS-step-5}. The following lemma shows that they are optimal online policies subject to ``shifting'' values of elements. We defer the proof of Lemma~\ref{lem:optimal-small-bin} to the appendix.

\begin{lemma}\label{lem:optimal-small-bin}
    Let $B \in \smallbin$ be a maximal small bin. Then there exists $\lambda^*$ such that the online policy for $B$ computed in Step~\ref{alg:PTAS-step-5} is an optimal online policy for $B$ with values of $u_t$, $t\in B$ being $v_t-\lambda^*$.
\end{lemma}

\subsubsection{Properties of Optimal Online Policies on Left-to-Right \lmbs}
Now that we know that the online policies computed in Step~\ref{alg:PTAS-step-5} are optimal online policies for some left-to-right \lmbs, let us understand the gains and behaviour of such optimal online policies. To study an optimal policy for a left-to-right \lmbs instance $\widetilde \mat = (\widetilde U, \widetilde \lam,  \widetilde c)$, let us consider the natural dynamic program for finding such policies.
For simplicity, each considered state consists of the current timestamp $t$ and the set $S\subseteq \widetilde U$ of currently selected elements.
Let $D_t(S)$ indicate the expected gain of the optimal pricing policy from items $\widetilde u_t$, $\widetilde u_{t+1}$, \ldots, $ \widetilde u_n$ given that we have currently selected items $S \subseteq \{\widetilde u_1, \ldots, \widetilde u_{t-1}\}$.
These values of $D_t(S)$ can be defined using a dynamic programming approach as follows
\[
    D_t(S) := 
    \begin{cases}
        D_{t+1}(S) \quad &\text{if $S \cup \{\widetilde u_t\}$ is infeasible for }\widetilde \mat \\
        \expect_{\widetilde v_t}\left[\max(D_{t+1}(S), D_{t+1}(S \cup \{\widetilde u_t\}) + \widetilde v_t)\right] \quad &\text{otherwise}
    \end{cases}
\]
together with
\[
    D_{n+1}(S) := 0\,.
\]
Hence, in the case that an optimal online policy reached the state given by $S\subseteq \widetilde U$ and $t$, then this optimal policy selects an element $\widetilde u_t$ with value $\widetilde v_t$ if and only if $S \cup \{\widetilde u_t\}$ is feasible and
\[
   \widetilde  v_t \geq D_{t+1}(S) - D_{t+1}(S \cup \{\widetilde u_t\})\,.
\]
In other words, $D_{t+1}(S) - D_{t+1}(S \cup \{\widetilde u_t\})$ is the threshold value at which the optimal online policy selects $\widetilde u_t$.

The following lemma shows that slight changes in the current state set $S$ lead only to slight changes in the number of elements selected by the optimal policy in the future.

\begin{lemma} \label{lem:firstuseless}
    Let $\widetilde \mat = (\widetilde U, \widetilde \lam,  \widetilde c)$ be a left-to-right \lmbs instance with ordering of elements $\widetilde u_1$, $\widetilde u_2$, \ldots, $\widetilde u_n$.
     Consider the selections done by an optimal online policy when the policy encounters a fixed realization of values $\widetilde v_2$, $\widetilde v_3$, \ldots, $\widetilde v_n$. For this fixed realization, let $\mu_0$ and $\mu_1$ be the number of elements selected by an  optimal online policy from $\{\widetilde u_2, \ldots, \widetilde u_n\}$ starting at $t=2$ with $S$ being $\varnothing$ and $\{\widetilde u_1\}$, respectively.
    Then we have that
    \[
        \mu_1 \leq \mu_0 \leq \mu_1 + 1\,.
    \]
\end{lemma}

To show Lemma~\ref{lem:firstuseless}, we first need the following lemma.
\begin{lemma}
    Let $\widetilde \mat = (\widetilde U, \widetilde \lam,  \widetilde c)$ be a left-to-right \lmbs instance with ordering of elements $\widetilde u_1$,~\ldots,~$\widetilde u_n$. Let $S$, $I$, $I' \subseteq \{\widetilde u_1, \ldots, \widetilde u_{t-1}\}$ be pairwise disjoint sets. Let $I \cup I'$ be the set $\{i_1, i_2, \ldots, i_k\}$, where  $i_1$, $i_2$, \ldots, $i_k$ is the arrival order of elements from $I \cup I'$. Let us define two subsets of $I \cup I'$ using the parity of their indices
    \[I_{odd} := \{i_1, i_3, \ldots\}\qquad \text{and}\qquad I_{even} := \{i_2, i_4, \ldots\}\,.\]
    If $S \cup I, S \cup I', S \cup I_{odd},$ and $S \cup I_{even}$ are feasible with respect to $\widetilde \mat$ then we have
    \begin{equation}\label{eq:odd-even}
        D_t(S \cup I) + D_t(S \cup I') \leq D_t(S \cup I_{odd}) + D_t(S \cup I_{even})\,,
    \end{equation}
    for all $t$.
\end{lemma}

\begin{proof}
    We show~\eqref{eq:odd-even} by induction on $t$.
    When $t = n+1$, the statement is true as $D_{n+1}(S) = 0$ for all $S \subseteq U$.
    Suppose that~\eqref{eq:odd-even}  holds for $t+1$.

    It is crucial for us to understand the thresholds that the value $v_t$ of $u_t$ should ``pass'' in order to be selected given that the current state is $S'$. These thresholds have the following structure.
    \[
        T_t(S'):=\begin{cases}
            D_{t+1}(S') - D_{t+1}(S' \cup \{u_t\}) &\text{if } S' \cup \{u_t\} \text{ is feasible for }\widetilde M\\
            +\infty &\text{otherwise}\,.
        \end{cases}
    \]
    Let us fix the value of $\widetilde v_t$ and define the following values for $S'\subseteq \{\widetilde u_1,\ldots, \widetilde u_{t-1}\}$.
\[
    F_t(S') := 
    \begin{cases}
        D_{t+1}(S') \quad &\text{if $S' \cup \{\widetilde u_t\}$ is infeasible for }\widetilde \mat \\
        \max(D_{t+1}(S'), D_{t+1}(S' \cup \{\widetilde u_t\}) + \widetilde v_t) \quad &\text{otherwise}
    \end{cases}
\]
Instead of showing~\eqref{eq:odd-even}, we show that the following inequality holds 
   \begin{equation}\label{eq:odd-even-conditioned}
        F_t(S \cup I) + F_t(S \cup I') \leq F_t(S \cup I_{odd}) + F_t(S \cup I_{even})\,.
    \end{equation}
By linearity and by the definition of $D_t(\cdot)$, the inequality~\eqref{eq:odd-even-conditioned} implies the inequality~\eqref{eq:odd-even}.   We consider three cases based on the behaviour of  $\widetilde v_t$ with respect to the thresholds $T_t(S \cup I)$ and $T_t(S \cup I')$.

    \textbf{Case $\widetilde v_t < T_t(S \cup I)$ and $\widetilde v_t < T_t(S \cup I')$.}
    Then $\widetilde u_t$ is selected at neither of these states, so we have that
    \[
        F_t(S \cup I) + F_t(S \cup I') = D_{t+1}(S \cup I) + D_{t+1}(S \cup I')\,,
    \]
    and by the definition of $D_t(\cdot)$ and $F_t(\cdot)$ we have
    \[
        F_t(S \cup I_{odd}) + F_t(S \cup I_{even}) \geq D_{t+1}(S \cup I_{odd}) + D_{t+1}(S \cup I_{even})\,,
    \]
    proving that the inductive step is valid.

    \textbf{Case $\widetilde v_t \geq T_t(S \cup I)$ and $\widetilde v_t \geq T_t(S \cup I')$.}
    Then $\widetilde u_t$ is accepted at both of these states, so we have that
    \[
        F_t(S \cup I) + F_t(S \cup I') = D_{t+1}(S \cup I\cup \{\widetilde u_t\}) + D_{t+1}(S \cup I'\cup \{\widetilde u_t\})+2 \widetilde v_t\,,
    \]
    Note, that for every bin $B\in\widetilde \lam$ due to the left-to-right arrival order and the definition of $I_{odd}$ and $I_{even}$ we have that the elements from $B$ in $I_{odd} \cup I_{even}$ are evenly distributed:
    \[
        |B \cap I_{even}| - 1 \leq |B \cap I_{odd}| \leq |B \cap I_{even}| + 1.
    \]
    This implies that
    \[
        \max\{|B\cap I_{odd}|, |B\cap I_{even}|\}\leq \max\{|B\cap I|, |B\cap I'|\}\,.
    \]
    Thus, since both sets $S \cup I\cup \{\widetilde u_t\}$ and $S \cup I'\cup \{\widetilde u_t\}$ are feasible with respect to $\widetilde M$, the sets $S \cup I_{odd}\cup \{\widetilde u_t\}$ and $S \cup I_{even}\cup \{\widetilde u_t\}$ are feasible with respect to $\widetilde M$.
    So by the inductive hypothesis and definition of $D_t(\cdot)$ and $F_t(\cdot)$, we have
    \begin{align*}
        F_t(S \cup I_{odd}) + F_t(S \cup I_{even}) &\geq D_{t+1}(S \cup \{\widetilde u_t\} \cup I_{odd}) + D_{t+1}(S \cup \{\widetilde u_t\} \cup I_{even}) + 2\widetilde v_t\\
        &\geq D_{t+1}(S \cup \{\widetilde u_t\} \cup I) + D_{t+1}(S \cup \{\widetilde u_t\} \cup I') + 2 \widetilde v_t\\
        &= F_t(S \cup I) + F_t(S \cup I')\,.
    \end{align*}

    \textbf{Case  when exactly one of $\widetilde v_t \geq T_t(S \cup I)$ or $\widetilde v_t \geq T_t(S \cup I')$ holds.}
    Suppose without loss of generality that $\widetilde v_t \geq T_t(S \cup I)$ holds. We have two cases depending on the parity of $k$ in the statement of the lemma. Let us suppose further that $k$ is even; the case when $k$ is odd is completely analogous. 
    Analogously to the above cases, both $S\cup I_{odd}\cup \{\widetilde u_t\}$ and $S\cup I_{even}$ are feasible with respect to $\widetilde \mat$ and we obtain 
    \begin{align*}
        F_t(S \cup I_{odd}) + F_t(S \cup I_{even}) \geq F_t(S \cup I) + F_t(S \cup I')\,.
    \end{align*}
\end{proof}

\begin{proof}[Proof of Lemma~\ref{lem:firstuseless}]
Let $J$ and $J'$ be the elements selected by the optimal online policy starting at $t=2$ with $S=\varnothing$ and $S=\{\widetilde u_1\}$, respectively.
Let $S:=J\cap J'$, $I:=J\setminus J'$ and $I':=J'\setminus J$ and let $I_{odd}$, $I_{even}$ be defined as in Lemma~\ref{lem:firstuseless}.

To prove the statement of the lemma we show that the following properties hold
\begin{align} \tag{ALT}\label{eq:alternate}
    I'_t:=I'\cap \{\widetilde u_1,\ldots, \widetilde u_t\}=I_{odd}\cap \{\widetilde u_1,\ldots, \widetilde u_t\}\\ \nonumber
    I_t:=I\cap \{\widetilde u_1,\ldots, \widetilde u_t\}=I_{even}\cap \{\widetilde u_1,\ldots, \widetilde u_t\}
\end{align}
for all $t=1,\ldots,n$. The property~\eqref{eq:alternate} shows that the behaviour of the optimal online policies ``alternates'' starting from $S=\varnothing$ and $S=\{\widetilde u_1\}$ after the first timestamp. Thus,~\eqref{eq:alternate} immediately implies the statement of the lemma. 

Let us prove the statement by induction on $t$. For $t=1$ the statement is trivial. Let us assume that~\eqref{eq:alternate} holds for the timestamp $t-1$ and let us prove that it holds also for $t$. We have two possible cases:  $|I'_{t-1}|=|I_{t-1}|+1$ and $|I'_{t-1}|=|I_{t-1}|$. Let us consider only the case $|I'_{t-1}|=|I_{t-1}|+1$ as the other case is completely analogous.   

Let $S_{t-1}$ be defined as $S\cap  \{\widetilde u_1,\ldots, \widetilde u_{t-1}\}$. Consider the sets $S_{t-1}\cup I_{t-1} \cup \{\widetilde u_t\}$ and $S_{t-1}\cup I'_{t-1} \cup \{\widetilde u_t\}$. Note that with respect to $\widetilde M$  they are either both infeasible, only the latter is infeasible, or they are both feasible. For the inductive step, the only non trivial case is when they are both feasible with respect to $\widetilde M$.  The threshold for selecting $\widetilde u_{t}$ in the case that the current state corresponds to $S_{t-1}\cup I_{t-1}$ is
    \[
        D_{t+1}(S_{t-1}\cup I_{t-1}) - D_{t+1}(S_{t-1}\cup I_{t-1} \cup \{\widetilde u_t\})
    \]
and in the case that the current state corresponds to $S_{t-1}\cup I'_{t-1}$, the threshold is 
    \[
        D_{t+1}(S_{t-1}\cup I'_{t-1}) - D_{t+1}(S_{t-1}\cup I'_{t-1} \cup \{\widetilde u_t\})\,.
    \]
By Lemma~\ref{lem:firstuseless} and by~\eqref{eq:alternate} for the timestamp $t-1$, we have
    \[
        D_{t+1}(S_{t-1}\cup I_{t-1}\cup \{\widetilde u_t\}) + D_{t+1}(S_{t-1}\cup I'_{t-1}) \geq D_{t+1}(S_{t-1}\cup I'_{t-1}\cup \{\widetilde u_t\}) + D_{t+1}(S_{t-1}\cup I_{t-1})\,.
    \]
 Thus the threshold for selecting $\widetilde u_{t}$, when the current state corresponds to $S_{t-1}\cup I_{t-1}$, is not larger than the same threshold when the current state corresponds to $S_{t-1}\cup I'_{t-1}$. Thus, $\widetilde u_t$ has passed both thresholds or only the threshold corresponding to $S_{t-1}\cup I_{t-1}$, showing~\eqref{eq:alternate} for the timestamp $t$. In both cases,~\eqref{eq:alternate} holds for the timestamp $t$. 
\end{proof}

\subsubsection{Results to Prove Concentration} \label{sec:concentration}

To show Lemma~\ref{lem:failure-prob}, we first need to show that the number of elements selected from each maximal small bin is very likely to be close to the expected number of elements it is supposed to select. The next lemma is the key result for proving Lemma~\ref{lem:failure-prob}, since it allows us to do an estimation of Chernoff's type in Lemma~\ref{lem:failure-prob}. The proof of it heavily relies on  Lemma~\ref{lem:firstuseless}. Lemma~\ref{lem:concentration} allows us to obtain several estimates for the number of elements suggested to be selected from $B$ in Step~\ref{alg:PTAS-step-5} as if the suggestions for elements were negatively correlated.

\begin{proof}[Proof of Lemma~\ref{lem:concentration}] In the proof, we use only the fact that the online policies computed in Step~\ref{alg:PTAS-step-5} are optimal online policies for $B$. Thus, these online policies have to satisfy Lemma~\ref{lem:firstuseless}.
We show the claim by induction on the total number  $k$ of elements in the bin $B$.
    The statement trivially holds when $k = 0$.
    Suppose that the statement is true for all instances with $k - 1$ elements. Abusing the notation, let $u_1$ be the first element in $B$ to arrive.
    Define $\mu := \expect[\sum_{u_t\in B} X_t]$, $\mu_1 :=~\expect[\sum_{u_t\in B,\, t\neq 1} X_t\,|\, X_1=1]$, $\mu_0 := \expect[\sum_{u_t\in B,\, t\neq 1} X_t\,|\, X_1=0]$ and $p_1=\prob[X_1=1]$, $p_0=\prob[X_1=0]$.
    
    By the inductive hypothesis, we have
    \begin{align*}
        \expect[e^{\alpha X_B}] = p_1 \cdot e^\alpha \cdot \expect[e^{\sum_{i=2}^n \alpha X_i}\, |\, X_1 = 1] + p_0 \cdot \expect[e^{\sum_{i=2}^n \alpha X_i}\, |\, X_1 = 0] \\\leq p_1 \cdot e^{\alpha} \cdot e^{(e^{\alpha}-1) \cdot \mu_1} + p_0 \cdot e^{(e^\alpha - 1) \cdot \mu_0}\,.
   \end{align*}
Consider the following program where $\mu$ is considered a constant, while $p_0$, $p_1$, $\mu_0$, and $\mu_1$ are considered variables.
    \begin{align*}
    &\text{maximize} & \quad & p_1 \cdot e^{\alpha} \cdot e^{(e^{\alpha}-1) \cdot \mu_1} + p_0 \cdot e^{(e^\alpha - 1) \cdot \mu_0} \\
    &\text{subject to} && p_0 + p_1 = 1,\qquad 0 \leq p_0 \leq 1 \\
    &&& p_1 \cdot \mu_1 + p_1 + p_0 \cdot \mu_0 = \mu, \qquad \mu_1 \leq \mu_0 \leq \mu_1 + 1
    \end{align*}
    Note that the last inequality holds due to Lemma~\ref{lem:firstuseless}.
    To prove the current lemma, we show that the objective value of this program is bounded from above by $e^{(e^\alpha - 1) \cdot \mu}$. 
    Let $p_0^*, p_1^*, \mu_0^*, \mu_1^*$ be a feasible solution of this program. 

    \textbf{Case  $p_1^* = 1$.}
    Then we must have that $\mu = \mu_1^* + 1$, and so the objective value of such a solution becomes $
        e^\alpha \cdot e^{(e^\alpha - 1) \cdot (\mu - 1)} \leq e^{(e^\alpha - 1) \cdot \mu}$,
    where the inequality follows from $\alpha \leq e^{\alpha - 1}$.

    \textbf{Case  $p_1^* = 0$.}
    Then the objective value of such a solution is exactly~$e^{(e^\alpha - 1) \cdot \mu}$.

    \textbf{Case $p_1^* \in (0, 1)$.} We obtain the following equation $
        \mu_1^* = (\mu - \mu_0^* p_0^* - p_1^*)/p_1^*$.
    From now on we can consider $p_0^*$, $p_1^*$ as constants, and so the program now corresponds to an optimization over a single variable $\mu_0^*$, and further, the objective function is convex in this variable.
    Therefore, it suffices to check that $\mu_0^*$ on the boundary of the arising feasibility region corresponds to an objective value that is at most $e^{(e^\alpha - 1) \cdot \mu}$.

    If  $\mu_1^* + 1 = \mu_0^*$, then $\mu_0^* = \mu$, and the objective value is
    \[
        p_1^* \cdot e^\alpha \cdot e^{(e^\alpha - 1) \cdot (\mu - 1)} + p_0^* \cdot e^{(e^\alpha - 1) \cdot \mu} \leq p_1^* \cdot e^{(e^\alpha - 1) \cdot \mu} + p_0^* \cdot e^{(e^\alpha - 1) \cdot \mu} = e^{(e^\alpha - 1) \cdot \mu},
    \] 
    where as before, we use the inequality $\alpha \leq e^{\alpha - 1}$.

    If $\mu_1^* = \mu_0^*$, then $\mu_1^* + p_1^*  = \mu$, and the objective value is
    \[
        e^{(e^\alpha  - 1) \cdot (\mu - p_1^*)} \cdot (p_1^* \cdot e^\alpha + p_0^*) \leq e^{(e^\alpha - 1) \cdot \mu},
    \]
    using the fact that 
   $
        p_1^* \cdot e^\alpha + p_0^* = p_1^* \cdot e^\alpha + (1 - p_1^*) \leq e^{(e^\alpha - 1) \cdot p_1}
    $.
    
\end{proof}

\subsubsection{Putting Everything Together}

\begin{theorem} \label{thm:left-to-right-lmbs}
    Algorithm~\ref{alg:PTAS} is a PTAS for left-to-right \lmbs instances when $K = \epsilon^{-4}$.
\end{theorem}
\begin{proof}

First, note that the set of elements selected by Algorithm~\ref{alg:PTAS} is independent with respect to $\mat'$, and so by Lemma~\ref{lem:logdepth} is independent with respect to $\mat$.

The online policies that are extracted in Step~\ref{alg:PTAS-step-5} of Algorithm~\ref{alg:PTAS} correspond to running online policies independently on maximal small bins. When run on maximal small bins independently, these online policies guarantee the gain  $LP_{\mat''}$, which is at least  $(1-\epsilon)^2 \cdot OPT_{\mat}$ by Lemma~\ref{lem:preprocessed-matroid-gain}. 

However, in Step~\ref{alg:PTAS-step-6} the online policies from Step~\ref{alg:PTAS-step-5} can fail to select some of the elements only beacuse of the feasibility constraints. By Lemma~\ref{lem:failure-prob} and since $K=\epsilon^{-4}$, for each $u_t\in U$ the probability that upon the arrival of $u_t$ the online policies suggest to select $u_t$ but it is not not feasible is at most~$3\epsilon$. Thus, the expected gain obtained by Algorithm~\ref{alg:PTAS} is at least $(1-3\epsilon)LP_{\mat''}$, and so at least $(1-\epsilon)^2(1-3\epsilon) \cdot OPT_{\mat}$.

\end{proof}

\subsection{Extending to Arrivals ``Close'' to Left-To-Right}

In this section, we extend the results for left-to-right \lmbs instances to \lmbs instances that are ``close'' to being left-to-right.
    Let $\mat = (U, \lam, c)$ be an \lmbs instance and $B \in \lam$ a bin.
    We say that $B$ is a \emph{left-to-right bin} if the instance $\mat$ restricted to $B$ is a left-to-right \lmbs instance.

We are able to extend the results because we do not need every bin in laminar matroid to be left-to-right. For our analysis to go through,
it is enough for every maximal small bin to be left-to-right.
Given the above observation, we obtain the following theorem by using a modified definition of small and  big bins, which is inspired by~\cite{AnariLMBS}.

\begin{theorem} \label{thm:constantcutoff}
    Let  $L \in \mathbb{N}$ be a constant. There is a PTAS for  \lmbs instances where every element lies in at most $L$ bins that are not left-to-right.
\end{theorem}

\begin{proof}[Proof Sketch]
    In \cite{AnariLMBS},  a PTAS is given for constant-depth \lmbs instances. For their algorithm, they have a different definition of  ``small'' and ``big'' bins. In Algorithm~\ref{alg:PTAS} we use one threshold $K=\epsilon^{-4}$ for all bins: a bin $A$ is big for our Algorithm~\ref{alg:PTAS} if $c(A)\geq K$. In \cite{AnariLMBS}, there are several thresholds: a bin $A$ is ``big'' for them if $c(A)\geq 1 / \delta^{L - d}$, where $\delta = \epsilon^2 / \log(1/\epsilon)$, $d$ is the depth of $A$, and $L$ is the depth of the whole matroid.
    The benefit of the approach from \cite{AnariLMBS} is that for any ``big'' bin $B$ and maximal ``small'' bin $A \subseteq B$, we have a separation between $c(A)$ and $c(B)$, in particular we have $c(A) \leq \delta \cdot c(B)$.
    This large separation in capacities removes the need to analyze the concentration for the number of  elements selected from $A$. In particular, there is no need for $A$ to be a left-to-right bin.
    Constant depth is needed to bound the capacity of any maximal ``small'' bin by a constant $1/\delta^L$ in order to explicitly solve the corresponding dynamic program.

    To generalize this to our setting, we use the following thresholds to define ``big'' and ``small'' bins.
    If the depth $d$ of a bin $A$ satisfies $d \leq L$, then we set the threshold for $A$ to be $\epsilon^{-4}/\delta^{L - d}$. 
    Otherwise, we set the threshold for $A$ to be $\epsilon^{-4}$.
    If a maximal ``small'' bin is identified before depth $L$, then it will necessarily have the factor of $\delta$ separation with any ``big'' bins containing it.
    Otherwise, we know that the identified maximal ``small'' bin is a left-to-right bin.
\end{proof}

\begin{corollary}[Theorem 3.5 from \cite{AnariLMBS}]
    There is a PTAS for constant-depth \lmbs instances.
\end{corollary}

The production constrained Bayesian selection problem from \cite{AnariLMBS} consists of $\lmbs$ instances in which all bins except the largest are left-to-right, plus some additional structure. This gives us the following corollary with $L = 1$.

\begin{corollary}[Theorem 2.3 from \cite{AnariLMBS}]
    There is a PTAS for production constrained Bayesian selection.
\end{corollary}

\section{Graphic Bayesian Selection}\label{sec:PSPACE}

\lmbs can be viewed as a special case of the more general Matroid Bayesian Online Selection problem, where instead of a laminar matroid, we are given any matroid $\mat$.
Gupta questioned in 2017 whether the best strategy for Matroid Bayesian Online Selection is computationally hard to find~\cite{Gupta}. 
We answer this question in the affirmative, even for the special case of graphic matroids.

\begin{theorem} \label{thm:gmbshard}
    There is an absolute constant $\alpha \in (0, 1)$ such that it is \ps-hard to approximate Graphic Matroid Bayesian Selection (\gmbs) to a factor of $\alpha$.
\end{theorem}

Note that we can compute the optimal strategy for \gmbs in polynomial space with a simple brute-force recursive algorithm, so \gmbs is in fact \ps-complete.

Also note that Theorem~\ref{thm:gmbshard} says it is \ps-hard to approximate the \textit{expected value} of an optimal strategy. 
This does not immediately imply that it is \ps-hard to act according to an approximate optimal strategy, but the proof of Theorem~\ref{thm:gmbshard} shows this as well.

\begin{definition}
    We say that a class of matroids $\mathscr{M}$ \emph{admits a PTAS-compatible distribution-agnostic order} if there exists an order $\sigma(\mat)$ for every matroid $\mat \in \mathscr{M}$ such that there is a Polynomial-Time Approximation Scheme for instances of Matroid Bayesian Online Selection with matroids from $\mathscr{M}$ and arrival orders given by $\sigma$.
\end{definition}

We show the following result by demonstrating that the hardness construction of Theorem~\ref{thm:gmbshard} can be embedded in any arrival order of sufficiently large complete graphs.
\begin{theorem} \label{thm:allordershard}
    The set of graphic matroids does not admit a PTAS-compatible distribution-agnostic order unless $\ps = \texttt{P}$.
\end{theorem}
As a corollary of Theorem~\ref{thm:left-to-right-lmbs}, we have that the set of laminar matroids admits a PTAS-compatible distribution-agnostic order by taking $\sigma(\mat)$ to be any left-to-right ordering on $\mat$.
Theorem \ref{thm:allordershard} provides a separation between graphic and laminar matroids.

\subsection{Hardness for Graphic Bayesian Selection}

\subsubsection{\ps-Hard Problem: \stwosat}

To prove Theorem \ref{thm:gmbshard}, we first need some results about Stochastic Max 2SAT.

\begin{definition}
In the \emph{Stochastic Max 2SAT problem}, henceforth referred to as \stwosat, the input is a 2CNF formula $\phi$ over an ordered list of variables $(x_1, x_2, \ldots, x_n)$, where $n$ is even and for every $i = 1, \ldots, n-1$, either $x_i$ or $x_{i+1}$ is contained in some clause of $\phi$. We choose a value of \texttt{True} or \texttt{False} for $x_1$. Then, nature sets $x_2$ to either \texttt{True} or \texttt{False} with a probability of $0.5$. We then get to choose the value of $x_3$, nature sets the value of $x_4$ to either \texttt{True} or \texttt{False}, and so on. Our goal is to maximize the expected number of satisfied clauses in $\phi$ after all the variables have been assigned a value. A variable is called a \emph{random} variable if it is set by nature, and is called \emph{deterministic} otherwise. 
\end{definition}

The following lemma is crucial for our reduction, since it provides a problem to reduce from. The proof of Lemma~\ref{lem:m2sathard} follows relatively directly from the analogous result for MAX-3SAT~\cite{Papadimitriou2021} and is included in the appendix.

\begin{lemma} \label{lem:m2sathard}
There exist absolute constants $k \in \mathbb{N}$ and $\alpha \in (0, 1)$ so that it is \ps-hard to compute an $\alpha$-approximation for a \stwosat instance $\phi$ satisfying the requirement that each variable appears in at most $k$ clauses of $\phi$.
\end{lemma}

\subsubsection{Reduction from \stwosat to \gmbs}\label{sec:reduction-2SAT}

Now, we can state our reduction from \stwosat to \gmbs.
Let $\phi$ be an instance of \stwosat as in Lemma \ref{lem:m2sathard} with variables $(x_1, x_2, \ldots, x_n)$ and constant~$k$. Let $m$ indicate the number of clauses in $\phi$. We construct an instance $\newinst$ of \gmbs as follows.

The vertex set in the graph for $\newinst$ consists of a single central vertex $w$ and two vertices for each variable $x_i$, $i=1,\ldots,n$ in $\phi$. We label the vertices for the variable $x_i$, $i=1,\ldots,n$ as $x_i$ and $\neg x_i$
\[
    V := \{w\} \cup \{\,x_i,\, \neg x_i\, :\, 1 \leq i \leq n \}\,.
\]

The edges of $\newinst$ arrive in three distinct phases and have distinct value distributions. Before specifying $\newinst$ in full detail, let us explain the intuition behind the construction. The first phase simulates the selection of truth values for variables in the same order as they appear in $\phi$. The second phase accounts for the number of clauses satisfied by the value selection of phase one. For this, the edges in the second phase correspond to clauses and attain very high values with very small probability, so that multiple clause edges attaining value is very unlikely. The third phase guarantees that for each variable at most one truth value is selected in the first phase by the optimal online policy.

\textbf{First Phase.}  Respecting the order of variables as in $\phi$, two edges $w  x_i$ and  $w \neg x_i$ arrive  for each variable $x_i$, $i=1,\ldots,n$. The edge $w  x_i$ arrives immediately before $w \neg x_i$.
    If  $x_i$ is a deterministic variable, both edges $w x_i$ and $w \neg x_i$ have deterministic value $1$. If $x_i$ is a random variable $x_i$, the edge $w x_i$ has value $2$ with probability $0.5$ and value $0$ otherwise; the edge $w \neg x_i$  has deterministic value~$1$.

\textbf{Second Phase.} An edge $\neg \ell_1 \neg \ell_2$ arrives for each clause $(\ell_1 \vee \ell_2)$ in $\phi$. Each such edge has value $m^4/2k$ with probability $m^{-4}$, and value $0$ otherwise. For each clause $(\ell)$ in $\phi$, i.e., for each clause with a single literal, an edge $w \neg \ell$ arrives with the same value distribution.

\textbf{Third Phase.} An edge $x_i\neg x_i$ arrives  for each variable $x_i$, $i=1,\ldots,n$. The edge $x_i\neg x_i$ has deterministic value $2$.

Note that within the first phase, the order of arrivals is important. In contrast, the order of edge arrivals within the second or third phase can be arbitrary.

\subsubsection{Optimal Online Policy for $\newinst$}

Let $OPT_{on}(\phi)$ refer to the expected value of the optimal online algorithm for \stwosat on $\phi$.
Let $ALG_{opt}$ denote an optimal algorithm for \gmbs. We can assume that $ALG_{opt}$ selects only arrived edges. We say that an edge \emph{arrived} if the value of the edge is nonzero. We defer the proof of next two lemmas to the appendix.

\begin{lemma} \label{lem:onevar}
Algorithm $ALG_{opt}$ selects exactly one of $wx_i$, $w\neg x_i$ for each variable $x_i$.
\end{lemma}

The following lemma shows that $ALG_{opt}$ respects the random choice that nature made for the truth value of each random variable.
\begin{lemma}\label{lem:randcorrect}
    For a random variable $x_i$, algorithm $ALG_{opt}$ selects $w x_i$ if and only if $w x_i$ arrives. 
\end{lemma}

The following lemma shows that the optimal online policy  $ALG_{opt}$ for $\newinst$ behaves greedily after the first phase. This can be seen as the value of arrived edges is non-increasing after phase 1, so it is always optimal to select an arrived edge due to the greedy selection property of matroids.
\begin{lemma}
 \label{lem:greedyselection}
After the first phase, $ALG_{opt}$ selects every edge that arrives and  can be selected without introducing a cycle.
\end{lemma}

\begin{remark}
          Note that the proofs of the above lemmas, i.e. Lemma~\ref{lem:onevar}, Lemma~\ref{lem:randcorrect} and Lemma~\ref{lem:greedyselection}, can be adopted to any online policy $ALG_{aprx}$ that attempts to maximize the expected total value of the selected elements but is limited to polynomial time computations. Indeed, the proofs of Lemma~\ref{lem:onevar}, Lemma~\ref{lem:randcorrect} and Lemma~\ref{lem:greedyselection} show that if $ALG_{aprx}$ does not follow the rules outlined in these lemmas, one can compute an alternative online policy $\mathcal A$ that satisfies all the outlined rules and runs in polynomial time and achieves a strictly better expected total value of the selected items.   
\end{remark}

\subsubsection{Gains for $\newinst$}
Let us consider the gains that are achieved by $ALG_{opt}$ and can be achieved by an online policy $ALG_{aprx}$ that runs in polynomial time and tries to maximize the expected total value of the selected elements.  Due to the results in the previous section we can assume that both $ALG_{opt}$ and $ALG_{aprx}$ follow the rules below:
\begin{enumerate}
    \item~\label{rule:onevar} exactly one of $wx_i$, $w\neg x_i$ is selected for each variable $x_i$
    \item~\label{rule:randcorrect} $w x_i$ is selected for a random variable $x_i$ if and only if $w x_i$ arrives
    \item~\label{rule:greedyselection} after the first phase, every edge is selected whenever the edge arrives and  can be selected without introducing a cycle.
\end{enumerate}

\begin{lemma}\label{lem:gain-estimate}
    Let $\mathcal{A}$ be an online policy that follows rules~\eqref{rule:onevar}, ~\eqref{rule:randcorrect} and~\eqref{rule:greedyselection}. Let $P_{\mathcal A}(\newinst)$ be the probability that a randomly selected clause in~$\phi$ is satisfied if the \texttt{True} and \texttt{False} values for $x_i$, $i=1,\ldots, n$ were assigned according to the edges selected by $\mathcal A$ in the first phase.   In particular, we assign $x_i$ the value \texttt{True} if $\mathcal A$ selects $w x_i$, and we assign  \texttt{False} otherwise. Then, we have
    \begin{align*}
\mathcal{A} (\newinst)= 1.25 n + 
&2n \cdot \left(1 - m^{-4}\right)^m + \left(2n + P_{\mathcal A}(\newinst) \cdot \left(\frac{m^4}{2k} - 2\right)\right) \cdot m \cdot m^{-4} \cdot \left(1 - m^{-4}\right)^{m-1} +\delta_{\mathcal A}\,,
\end{align*}
where $\delta_{\mathcal A}$ is a number in $ [0, 2m^{-1}]$.
\end{lemma}
Clearly, $P_{\mathcal A}(\newinst)$ in Lemma~\ref{lem:gain-estimate} is the number of clauses satisfied by such an assignment divided by the total number of clauses, i.e. by $m$. 
We defer the proof of Lemma~\ref{lem:gain-estimate} to the appendix.

\subsubsection{Putting Everything Together}

Let us define $\gamma:=(1-m^{-4})^{m-1}$ and notice that the limit of $\gamma$  as $m$ goes to infinity equals~$1$. Recall that $k$ is an absolute constant. Thus, considering the expected gain of an online policy, we can conclude that when $\mathcal A$ is $ALG_{opt}$, the value of $P_{\mathcal A}$ is $OPT_{on}(\phi)/m$.
Let us now rewrite the expected gain of the optimal policy $ALG_{opt}$ as follows
\begin{align*}
\underbracket{1.25 n + 2n \cdot \gamma + 2n\cdot m^{-3}\cdot \gamma +\delta_{ALG_{opt}}}_{=:Q'}\,+\, \underbracket{OPT_{on}(\phi) \cdot \left(\frac{m^4}{2k} - 2\right) \cdot  m^{-4} \cdot \gamma}_{=:Q}\,.
\end{align*}
Let us define $Q$ and $Q'$ as above and define $\beta:=54k/\gamma$. 

\begin{lemma}[Fact 2.1 from \cite{Papadimitriou2021}] \label{apxfact}
Let $Q, Q' \geq 0$ be positive quantities, such that $\frac{Q'}{Q} \leq \beta$, and let $\alpha \in (0, 1)$. Then, an $(\frac{\alpha + \beta}{1 + \beta})$-approximation to $Q + Q'$ yields an $\alpha$-approximation to $Q$.
\end{lemma}

To use Fact~\ref{apxfact}, we need to show that $Q'/Q\leq \beta$. Notice that $n \leq 4m$ because no consecutive pair of variables is missing from all clauses. Also, $OPT_{on}(\phi) \geq \frac{m}{2}$ because we can achieve $\frac{m}{2}$ satisfied clauses with a random variable assignment. So we obtain $n \leq 8 \cdot OPT_{on}(\phi)$. Thus we have
\[ Q' \leq 3.25n + 0.25+ \delta_{ALG_{opt}} \leq 26\cdot OPT_{on}(\phi) + 0.5\leq 27 \cdot OPT_{on}(\phi) \]
for sufficiently large $m$. Now we have that
\[ Q'/Q \leq \left(27 \cdot OPT_{on}(\phi)\right)/\left(0.25\cdot \gamma\cdot OPT_{on}(\phi)/k\right) = \beta \]
for sufficiently large $m$.
Also, since $k$ is an absolute constant, we have that $\beta=108k/\gamma$ can be seen as a constant. So, we can apply Fact~\ref{apxfact} and find that an $\alpha' := \frac{\alpha + 108k/\gamma}{1 + 108k/\gamma}$ approximation for \gmbs gives an $\alpha$ approximation for \stwosat. Thus,  by Lemma~\ref{lem:m2sathard} there is some constant $\alpha' \in (0, 1)$ such that it is \ps-hard to approximate \gmbs to a factor of $\alpha'$.

\subsection{No PTAS-Compatible Distribution-Agnostic Order}
Let us show that for a graphic matroid on a complete graph there is no PTAS-compatible distribution-agnostic order. Let us be given a complete graph $G=(V,E)$ with $|V|=3n+2m+1$, where $n$ is even, and let us be given some arrival order for~$E$. Let us show how one can construct a reduction for \stwosat analogously to Section~\ref{sec:reduction-2SAT}.

\textbf{Edges for Third Phase.} In the reverse order of arrival, greedily construct a matching $A:=\{a_1,\ldots, a_n\}$ such that $|A|=n$. Each edge in $A$  has deterministic value $2$.

\textbf{Edges for Second Phase.} In the reverse order of arrival, after constructing $A$, greedily keep adding edges to $B:=\{b_1,\ldots, b_m\}$ such that $A\cup B$ is a matching and $|B|=m$. Each edge in $B$ has value $m^4/2k$ with probability $m^{-4}$, and value $0$ otherwise.

\textbf{Edges for First Phase.} Consider the vertex set $U$, which are the vertices not matched by the matching~$A\cup B$. Consider a vertex $w\in U$ and let $v_1$, $\overline{v}_1$, $v_2$, $\overline{v}_2$, \ldots, $v_n$, $\overline{v}_n$ be such that the edges $wv_1$, $w\overline{v}_1$, $wv_2$, $w\overline{v}_2$, \ldots, $wv_n$, $w\overline{v}_n$ are the first $n$ edges adjacent to $w$ and their arrivals are in the above order.  If  $i$ is odd, both edges $w v_i$ and $w \overline{v}_i$ have deterministic value $1$. If  $i$ is even, the edge  $w v_i$ has value $2$ with probability $0.5$  and value $0$ otherwise; and the edge $w \overline{v}_i$  has deterministic value $1$.

\textbf{Auxiliary Edges.} For each $i=1,\ldots, n$ both edges $v_i \alpha_i$ and $\overline{v}_i \gamma_i$ have deterministic value $3$, where  $\alpha_i$ and $\gamma_i$ are the endpoints of $a_i$. For each $j=1,\ldots, m$ both edges $u_j \beta_j$ and $l_j\tau_j$ have deterministic value $3$, where  $\beta_j$, $\tau_j$ are the endpoints of $b_j$ and $l_j$, $u_j$ correspond to ``negations'' of the literals in the $j$th clause of $\phi$.

Every remaining edge has deterministic value $0$. Note, that all auxiliary edges arrive before the edges in the second and third phase due to the greedy construction of $A$ and $B$. Also all auxiliary edges are selected by an optimal online policy.

The same analysis as in Section~\ref{sec:reduction-2SAT}, but with a slightly degraded inapproximability constant due to the additional guaranteed gain from the auxiliary edges.

\newpage

\bibliographystyle{splncs04}
\bibliography{mbs}

\newpage
\appendix

\section{Proofs from Section~\ref{sec:PTAS}}
\begin{proof}[Proof of Lemma~\ref{lem:optimal-small-bin}]
Let $\{\varx_t^{\star}(\vecs, v_t), \varx_t^{\star}(v_t), \vy^*_{t}(\vecs), \mathcal{N}^*_A\}$ indicate an optimal solution for~\eqref{lp2} as in Step~\ref{alg:PTAS-step-4}.  
Note that the part  of the optimal solution which corresponds to the elements from~$B$ is also an optimal solution for the following linear program.
\begin{align*}
&\text{maximize} & \quad & \sum_{u_t \in B} \expect_{v_t} [v_t \cdot \varx_t(v_t)] \\
&\text{subject to} &&\{ \varx_t(\vecs, v), \varx_t(v), \vy_t(\vecs)\} \in \mathcal{P}^B &  \\
&&& \sum_{u_t \in B} \expect_{v_t}[\varx_t(v_t)] = \mathcal{N}^{\star}_B & 
\end{align*}

Let us now consider the following relaxation of the above linear program by introducing $\lambda$.
\begin{align*}
&\text{maximize} & \quad & \sum_{u_t \in B} \expect_{v_t} [v_t \cdot \varx_t(v_t)] + \lambda (\mathcal{N}_B^* - \sum_{u_t \in B} \expect_{v_t}[\varx_t(v_t)]) \\
&\text{subject to} &&\{ \varx_t(\vecs, v), \varx_t(v), \vy_t(\vecs)\} \in \mathcal{P}^B & 
\end{align*}
There exists $\lambda^*$ such that for $\lambda=\lambda^*$ every optimal solution for the first linear program is an optimal solution for the second linear program.
Indeed, such a $\lambda^*$ can be obtained from the dual optimal solution for the first linear program.

Thus, the part  of the optimal solution which corresponds to the elements from~$B$ is also an optimal solution for the following linear program.
\begin{align*}
&\text{maximize} & \quad & \sum_{u_t \in B} \expect_{v_t} [(v_t - \lambda^{\star}) \cdot \varx_t(v_t)] \\
&\text{subject to} &&\{ \varx_t(\vecs, v), \varx_t(v), \vy_t(\vecs)\} \in \mathcal{P}^B &
\end{align*}
Note that the last linear program is encoding optimal online policies on the bin $B$, where all values for~$v_t$ are shifted by $\lambda^*$.   
\end{proof}

\section{Proofs from Section~\ref{sec:PSPACE}}
\begin{lemma}[Lemma 2.3 in \cite{Papadimitriou2021}] \label{lem:m3sathard}
There exist absolute constants $k \in \mathbb{N}$ and $\alpha \in (0, 1)$ so that it is \ps-hard to compute an $\alpha$-approximation for a \sthreesat instance $\phi$ satisfying the requirement that each variable appears in at most $k$ clauses of $\phi$.
\end{lemma}

\begin{proof}[Proof of Lemma \ref{lem:m2sathard}]
Let $\phi'$ be an instance of \sthreesat with variables $x_1$, $x_2$, \ldots, $x_n$ such that each variable appears in at most $k'$ clauses of $\phi'$. Let us construct an instance of \stwosat $\phi$ as follows:
\begin{itemize}
    \item define $x_1$, $x_2$, \ldots, $x_n$ as variables of $\phi$ 
    \item if a clause in $\phi'$ contains at most two literals, then directly add this clause to $\phi$
    \item if $i$th clause in $\phi'$ contains three literals, i.e. is of the form $(\ell_1 \vee \ell_2 \vee \ell_e)$, define a new deterministic variable $c_i$, append it to the end of our current variables; after that add ten following \emph{new clauses} to $\phi$
\[ (\ell_1), (\ell_2), (\ell_3), (c_i), (\neg \ell_1 \vee \neg \ell_2), (\neg \ell_2 \vee \neg \ell_3), (\neg \ell_1 \vee \neg \ell_3), (\ell_1 \vee \neg c_i), (\ell_2 \vee \neg c_i), (\ell_3 \vee \neg c_i) \]
\end{itemize}

First, note that $\phi$ satisfies the requirements of Lemma~\ref{lem:m2sathard} for $k=10 k'$. Indeed, each clause has at most two variables, and each variable appears in at most $k$ clauses of $\phi$. 

Consider any truth values assigned to literals $\ell_1, \ell_2, \ell_3$. If the values of $\ell_1$, $\ell_2$, $\ell_3$ satisfy $(\ell_1 \vee \ell_1 \vee \ell_3)$, then we can assign a value to $c_i$ so that $7$ of our new clauses are satisfied. Note that there is no assignment of \texttt{True} and \texttt{False} to $c_i$  satisfying more than $7$ new clauses. If the values of $\ell_1$, $\ell_2$, $\ell_3$ do not satisfy $(\ell_1 \vee \ell_2 \vee \ell_3)$, then we can satisfy at most $6$ of our new clauses by any assignment to $c_i$; and we can satisfy exactly $6$ by setting $c_i$ to \texttt{False}.

Let $OPT_{on}(\phi')$ denote the optimal expected value for $\phi'$, and equivalently for $OPT_{on}(\phi)$. We claim that $OPT_{on}(\phi) = 6m + OPT_{on}(\phi')$ where $m$ is the number of clauses with three literals in $\phi'$. 

We can construct an algorithm for $\phi$ that mimics the decisions of $OPT_{on}(\phi')$ until it reaches clause variables, at which point it assigns them to greedily maximize remaining clauses. As we saw above, we gain $6$ satisfied clauses in $\phi$ for each non-satisfied clause in $\phi'$ and $7$ satisfied clauses in $\phi$ for each satisfied clause in $\phi$, giving 
\[ OPT_{on}(\phi) \geq 6m + OPT_{on}(\phi')\,. \]

Now construct an algorithm for $\phi'$ that mimics the decisions of $OPT_{on}(\phi)$. An optimal algorithm for $\phi$ always yields value at least $6$ satisfied clauses per a three-literal clause in $\phi'$, since it can set the clause variables $c_i$, $i=1,\ldots,m$ to get at least $6$ satisfied clauses for each original three-literal clause of $\phi'$. 

Further, as we saw earlier, it is impossible for $OPT_{on}(\phi)$ to gain more than $7$ satisfied clauses for any given three-literal clause of $\phi'$. In particular, it can gain $7$ exactly when the values of  $x_1$, $x_2$, \ldots, $x_n$ satisfy the clause from $\phi'$. So our constructed algorithm for $\phi$ satisfies $OPT_{on}(\phi) - 6m$ clauses, which combined with above implies the equality
\[
    OPT_{on}(\phi) = 6m + OPT_{on}(\phi')\,.
\]

Now suppose that $\alpha' \in (0, 1)$ is the constant from Lemma \ref{lem:m3sathard}. In order to use Lemma~\ref{apxfact}, let $Q := OPT_{on}(\phi')$, $Q' := 6m$, and $\beta := 12$. Note that $OPT_{on}(\phi') \geq \frac{m}{2}$ and
\[
    \frac{Q'}{Q} \leq \frac{6m}{m/2} = \beta\,.
\]
Now Lemma~\ref{apxfact} implies an $\alpha := \frac{\alpha + 12}{13}$ approximation for \stwosat yields an $\alpha'$-approximation for \sthreesat. So it is \ps-hard to approximate \stwosat within a factor of $\alpha$.
\end{proof}

\begin{proof}[Proof of Lemma~\ref{lem:onevar}]
First, we show that $ALG_{opt}$ selects at least one of $wx_i$, $w\neg x_i$ for each variable $x_i$.  Suppose for the sake of contradiction that after some history occurring with probability $q > 0$, $ALG_{opt}$ selects neither $wx_i$ nor $w \neg x_i$. Let $\mathcal{A}$ be an algorithm that matches the decisions of $ALG_{opt}$, except that it selects $w \neg x_i$, and afterwards does not select any edges that would then introduce a cycle. 

Consider what we lose out on by running $\mathcal{A}$ instead of $ALG_{opt}$. Let $\mathcal{A}(\newinst)$ be the expected value obtained by $\mathcal{A}$ on $\newinst$. If no edge adjacent to $x_i$ or $\neg x_i$ arrive in the second phase, then the vertex $x_i$ has no adjacent selected edges before the third phase. So both $ALG_{opt}$ and $\mathcal{A}$ select the edge $x_i \neg x_i$ with value $2$ in the third phase, since in this case $x_i \neg x_i$ does not lead to a cycle for any of them. 

If some  edge adjacent to $x_i$ or $\neg x_i$ arrives in the second phase, then $\mathcal{A}$ can miss out on at most one edge that $ALG_{opt}$ selects. The value of this missed edge can be at most $m^4/2k$. Here, we assume $m^4/2k \geq 2$, which is true for sufficiently large $m$. By the union bound, the probability that any clause edge adjacent to $x_i$ or $\neg x_i$ arrives is at most $k \cdot m^{-4}$, so we have
\[
\mathcal{A}(\newinst) - OPT_{on}(\newinst) \geq q \left( 1 - k \cdot m^{-4} \cdot \frac{m^4}{2k}\right) = \frac{q}{2} > 0 \,,
\]
contradicting the optimality of $ALG_{opt}$. 

Second, note that $ALG_{opt}$ never selects both $w x_i$ and $w \neg x_i$ for any variable $x_i$, $i=1,\ldots, n$. To see this fact, observe that it is more beneficial for $ALG_{opt}$ to accept only $w x_i$ in the first phase and to commit to accept $x_i \neg x_i$ in the third phase.
\end{proof}

\begin{proof}[Proof of Lemma~\ref{lem:randcorrect}]

It is sufficient to show that if $w x_i$ arrives, i.e. has value $2$, $ALG_{opt}$ selects $w x_i$. For the sake of contradiction suppose that after some history occurring with probability $q > 0$, the edge $w x_i$ arrives but $ALG_{opt}$ does not select it.

Let $\mathcal{A}$ be an algorithm that matches the decisions of $ALG_{opt}$, except that it selects $w x_i$ when it arrives, and afterwards $\mathcal{A}$ does not select any edges that would then introduce a cycle. Analogously to the proof of Lemma~\ref{lem:onevar} we obtain
\[
\mathcal{A}(\newinst) - ALG_{opt}(\newinst) \geq q \left( 1 - k \cdot m^{-4} \cdot \frac{m^4}{2k}\right) = \frac{q}{2} > 0\,. 
\]
contradicting the optimality of $ALG_{opt}$. 
\end{proof}

\begin{proof}[Proof of Lemma~\ref{lem:gain-estimate}]
Let ${\mathcal A}_1(\newinst)$, ${\mathcal A}_2(\newinst)$ and ${\mathcal A}_3(\newinst)$ be the random variables indicating the total value of elements selected by $\mathcal A$ in each of the three phases. We have that $\expect[{\mathcal A}_1(\newinst)]=1.25n$ because in expectation each deterministic variable and randomized variable in $\phi$ lead to the gain of $1$ and $1.5$ in expectation by rules~\eqref{rule:onevar} and~\eqref{rule:randcorrect}.

Let $C$ be the random variable indicating the number of edges that arrived in the second phase. We estimate the gain from ${\mathcal A}_2(\newinst)$ and ${\mathcal A}_3(\newinst)$ conditioned on three possible events $C=0$, $C=1$, and $C>1$. Clearly, we have $\expect [{\mathcal A}_2(\newinst)+{\mathcal A}_3(\newinst)\,\mid\, C = 0] = 2n$, as under the condition $C=0$ there is no gain in the second phase and every edge is selected in the third phase by rule~\eqref{rule:greedyselection}. We also have
\[\expect [{\mathcal A}_2(\newinst)+{\mathcal A}_3(\newinst)\,\mid\, C = 1] = 2n + P_{\mathcal A}(\newinst) \cdot \left(\frac{m^4}{2k} - 2\right)\,,\]
because in this case the edge that arrived in the second phase is selected exactly when the corresponding clause is satisfied by the \texttt{True} and \texttt{False} assignment produced by the first phase, while the clause is satisfied with probability $P_{\mathcal A}(\newinst)$.

Let us state a trivial bound for the gain in the second and third phases conditioned on $C>2$
\[\expect [{\mathcal A}_2(\newinst)+{\mathcal A}_3(\newinst)\,\mid\, C > 1] \leq 2n + m \cdot \frac{m^4}{2k} \leq 2 \cdot \frac{m^5}{2k} \leq m^5\,,\]
which is obtained by assuming that all edges in the last two phases arrive and $\mathcal A$ selects all of them. However, the event that $C>1$ has a very small probability of occuring
\[
\prob[C > 1]= \sum_{t = 2}^m \binom{m}{t} \left(m^{-4}\right)^t \left(1-m^{-4}\right)^{m-t}\leq \sum_{t = 2}^m m^t \left(m^{-4t}\right)\leq 2m^{-6}\,.
\]
So we have
\[
\delta_{\mathcal A}:=\expect [{\mathcal A}_2(\newinst)+{\mathcal A}_3(\newinst)\,\mid\, C > 1] \cdot \prob[C > 1]\in [0, 2m^{-1}]\,.
\]
Thus, the expected gain of $\mathcal A$ can be estimated as follows
\begin{align*}
\mathcal{A} (\newinst)= \,& \expect[{\mathcal A}_1(\newinst)]+ \expect[{\mathcal A}_2(\newinst)+{\mathcal A}_3(\newinst)]=1.25 n \,+ \\
&2n \cdot \left(1 - m^{-4}\right)^m + \left(2n + P_{\mathcal A}(\newinst) \cdot \left(\frac{m^4}{2k} - 2\right)\right) \cdot m \cdot m^{-4} \cdot \left(1 - m^{-4}\right)^{m-1} +\delta_{\mathcal A}\,,
\end{align*}
where we used that $\prob[C=0]=\left(1 - m^{-4}\right)^m$ and $\prob[C=1]= m \cdot m^{-4} \cdot \left(1 - m^{-4}\right)^{m-1}$. 
\end{proof}

\section{Laminar Instance with Bad Concentration}

\begin{proof}[Proof of Theorem \ref{thm:badconcentration}]
     For simplicity, we will label the elements of the matroid as $U = \{u_1, u_2, \ldots, u_n\}$. For each element $u_i$, the produced instance uses only two-point distributions with support $\{0, v_i\}$ for some $v_i > 0$. Let $p_i$ be the probability that $u_i$ takes on value $v_i$. Thus, with probability $1 - p_i$ the element $u_i$ takes on value $0$.

Given natural numbers $r$ and $k$, the following algorithm generates an anti-concentrated instance $\mat$. Here, the choice of $k$ is later determined by  the value of $\varepsilon$.

    \begin{algorithm}[H]
        \caption{PRODUCE-INSTANCE(r, k)}
        \label{alg:anticoncentrated}
        $v \gets 10^{1+2k},\quad v' \gets 10^k,\quad p \gets 10^{-1-k}$ \;
        $n_1 \gets 1, \quad U_1 \gets \{u_1\}$ \;
        $v_1 \gets 1, \quad p_1 \gets \frac{1}{2}$ \;
        $\Delta_1 \gets 1$ \;
        $\lam_1 \gets \{\{u_1\}\}, \quad c_1(\{u_1\}) \gets 1$ \;
        \For{$i = 2, 3, \ldots, r$}{
            $\alpha_i \gets \frac{\Delta_{i-1}}{10 (r+1)^2} \cdot \frac{1}{v}$ \;
            $U_i \,\gets\, U_{i-1} \cup \{u_{n_{i-1}+1}, \ldots, u_{n_{i-1} + i + 2}\}, \quad n_i \gets n_{i-1} + i + 2$ \;
            $\lam_i\, \gets\, \left(\lam_{i-1} \setminus \{U_{i-1}\}\right) \,\cup\, \left\{U_{i-1} \cup \{u_{n_{i-1}+i+1}\},\, U_i, \{u_{n_{i-1}+1},\, u_{n_{i-1} + i + 2}\}\right\}$ \;
            $c_i(B)\, \gets\, c_{i-1}(B) \,\text{ for all } \, B \in \lam_{i-1} \cap \lam_i$ \;
            $c_i(U_{i-1} \cup \{u_{n_{i-1}+i+1}\}) \gets i-1, \quad c_i(U_i) \gets i,\quad c_i(\{u_{n_{i-1}+1}, u_{n_{i-1} + i + 2}\}) \gets 1$ \;
            $\Delta_i \gets \alpha_i \cdot v \cdot p^i \cdot (1 - p)$ \;
            $p_{n_{i-1} + 1} \gets 1,\quad v_{n_{i-1} + 1} \gets \alpha_i \cdot 10^k + \frac{\Delta_i}{2}$ \;
            $p_{n_{i-1} + 2},\, \ldots,\, p_{n_{i-1} + i + 2} \gets p$ \;
            $v_{n_{i-1} + 2},\, \ldots,\, v_{n_{i-1} + i} \gets \alpha_i \cdot v'$ \;
            $v_{n_{i-1} + i + 1},\quad v_{n_{i-1} + i + 2} \gets \alpha_i \cdot v$ \;
        }
        return $\mat = (U_r, \lam_r, c_r)$ with values $v_1, \ldots, v_{n_r}$ and probabilities $p_1, \ldots, p_{n_r}$ \;
    \end{algorithm}

    \begin{figure}[h]
    \centering
    \includegraphics[width=0.65\textwidth]{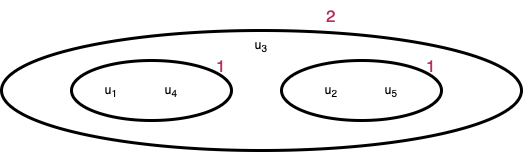}
    \caption{The laminar family, capacities, and arrival order generated by Algorithm~\ref{alg:anticoncentrated} when $r = 2$.}
    \end{figure}
    
    \begin{figure}[h]
    \centering
    \includegraphics[width=1\textwidth]{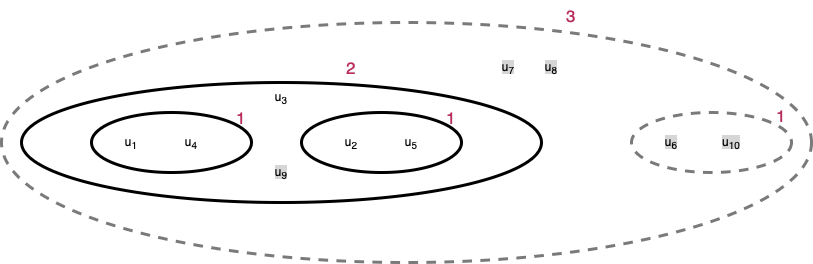}
    \caption{The laminar family, capacities, and arrival order generated by Algorithm~\ref{alg:anticoncentrated} when $r=3$. New sets in the laminar family are drawn with dashed lines, and new items are highlighted.}
    \end{figure}

    We note that the total number of items is $|U_r| = 1 + \sum_{i=2}^r i + 2 \leq (r+1)^2$.

\bigskip

    \begin{lemma} \label{lem:badconc1}
        For every $i \geq 2$,  the probability that we have $OPT(\mat) \cap \{u_1 \ldots, u_{n_i}\} = \varnothing$, conditioned on $OPT(\mat) \cap \{u_1 \ldots, u_{n_{i-1}}\} = \varnothing$,
        is at least $(1 - p)^{i+1}$.
    \end{lemma}
    \begin{proof}
        Let us re-label the elements so that $(w_1, w_2, \ldots, w_{i+2}) = (u_{n_{i-1} + 1}, u_{n_{i-1} + 2}, \ldots, u_{n_{i-1} + i + 2})$.
Let us show that  conditioned on $OPT(\mat) \cap \{u_1 \ldots, u_{n_{i-1}}\} = \varnothing$ the optimal policy does not select $w_1$. This suffices to show the statement of the lemma, since
        the probability that all elements among $w_2, \ldots, w_{i+2}$ have value zero is $(1 - p)^{i+1}$.

Now let us condition on $OPT(\mat) \cap \{u_1 \ldots, u_{n_{i-1}}\} = \varnothing$ and estimate the expected gain of strategies that select and do not select $w_1$.

\smallskip
        Let us consider a strategy that does not select $w_1$ and after that greedily  accepts every element among $w_2, \ldots, w_{i+2}$ that has non-zero value and is feasible.
        It is feasible to select every element among $w_2, \ldots, w_{i+2}$ that has non-zero value, unless all of $w_2, \ldots, w_{i+2}$ have non-zero values. On the elements $w_2$, \ldots $w_{i+2}$, this strategy yields expected value
        \[
            \alpha_i \cdot \left((i-1) \cdot v' \cdot p + 2 \cdot v \cdot p - v \cdot p^{i+1}\right)\,.
        \]
\smallskip
        
        Now consider any strategy that selects the element $w_1$. Clearly, such a strategy cannot select $w_{i+2}$ due to the constraints in the laminar matroid. On the elements $w_2$, \ldots $w_{i+2}$, this strategy yields expected value at most \[
            \alpha_i \cdot \left((i-1) \cdot v' \cdot p + v \cdot p - v \cdot p^i\right).
        \]
        Indeed, to maximize the expected gain on  $w_2$, \ldots $w_{i+1}$, this strategy  greedily selects any item with non-zero value.

        Note that so far we considered only gains on the elements $w_2$, \ldots $w_{i+2}$. However, there are elements with arrivals after the arrival of $w_{i+2}$. 
        Due to our scaling with $\alpha_{i+1}, \alpha_{i+2}, \ldots$, we know that the value of any element after $w_{i+2}$ is at most $\Delta_i/(10 (r+1)^2)$.
        There are at most $(r+1)^2$ elements arriving after $w_{i+2}$, so the sum of their value is at most $\Delta_i/10$.
        Therefore, if a strategy selects $w_1$, its expected value from elements arriving after $w_1$ is at most
        \[
            \alpha_i \cdot \left((i-1) \cdot v' \cdot p + v \cdot p - v \cdot p^i\right) + \frac{\Delta_i}{10}.
        \]

        So for the optimal strategy to select $w_1$, the value of $w_1$ should be at least
        \begin{align*}
            &\alpha_i \cdot \left((i-1) \cdot v' \cdot p + 2 \cdot v \cdot p - v \cdot p^{i+1}\right) - \alpha_i \cdot \left((i-1) \cdot v' \cdot p + v \cdot p - v \cdot p^i\right) + \frac{\Delta_i}{10} \\
            = &\alpha_i \cdot \left(v \cdot p  - v \cdot p^{i+1} + v \cdot p^i\right) - \frac{\Delta_i}{10} \\
            = &\alpha_i \cdot 10^k + \Delta_i - \frac{\Delta_i}{10} \\
            > &\alpha_i \cdot 10^k + \frac{\Delta_i}{2}.
        \end{align*}
Therefore, the optimal solution does not select $w_1$ because the value of $w_1$ is only $\alpha_i \cdot 10^k + \Delta_i/2$.      
    \end{proof}

\bigskip

    \begin{lemma} \label{lem:badconc2}
        For every $i \geq 2$,  the probability that we have $|OPT(\mat) \cap \{u_1 \ldots, u_{n_i}\}| = i$, conditioned on $|OPT(\mat) \cap \{u_1 \ldots, u_{n_{i-1}}\}| = i-1$,
        equals $1$.
    \end{lemma} 
    \begin{proof}
        Let us re-label the elements so that $(w_1, w_2, \ldots, w_{i+2}) = (u_{n_{i-1} + 1}, u_{n_{i-1} + 2}, \ldots, u_{n_{i-1} + i + 2})$.
        Let us show that conditioned on $|OPT(\mat) \cap \{u_1 \ldots, u_{n_{i-1}}\}| = i-1$ the optimal policy selects $w_1$.
        To do so, we analyze the performance of possible strategies on all elements after $w_{i+2}$ in both the situation that $w_1$ is selected and the situation that it is skipped.

        If $w_1$ is skipped by a strategy, then to obtain optimal gain on $w_2, \ldots, w_{i+2}$, a strategy has to greedily select the first item that has non-zero value, noting that $w_{i+1}$ cannot be taken as it is contained in the saturated constraint corresponding to $U_{i-1}\cup\{w_{i+1}\}$.
        The probability that at least one of $w_2$, \ldots, $w_i$ has a non-zero value is $(1 - (1 - p)^{i-1})$, in which case we gain value $\alpha_i \cdot v'$.
        Otherwise, with probability $(1 - p)^{i-1} \cdot p$, none of $w_2, \ldots, w_i$ have a non-zero value but $w_{i+2}$ does, and so a strategy gains the value $\alpha_i \cdot v$.
        So, a strategy skipping the element $w_1$ will gain in expectation on the elements $w_2 \ldots, w_{i+2}$ at most
        \[
            \alpha_i \cdot (v' \cdot (1 - (1 - p)^{i-1}) + v \cdot p \cdot (1 - p)^{i-1}) = \alpha_i \cdot 10^k.
        \]
        As seen in the proof of Lemma \ref{lem:badconc1}, the total value of all elements after $w_{i+2}$ is bounded by $\Delta_i/10$. So the expected gain of a strategy skipping $w_1$ is at most $\alpha_i \cdot 10^k + \Delta_i/10$. 

        Now, note that $\alpha_i \cdot 10^k + \Delta_i/10<\alpha_i \cdot 10^k + \Delta_i/2$ and $\alpha_i \cdot 10^k + \Delta_i/2$ is the value of $w_1$. Hence, the optimal strategy selects $w_1$, as wanted. 
        
    \end{proof}

    Finally, we note that item $u_1$ is taken by the optimal algorithm whenever its value is non-zero. So the optimal strategy selects $u_1$ with probability $1/2$. Applying Lemmas \ref{lem:badconc1} and \ref{lem:badconc2} inductively, this implies that 
    \[    
        \prob[\,|OPT(\mat)| = 0\,] \geq \frac{1}{2} \cdot (1 - p)^{(r+1)^2}
    \]
    and
    \[
        \prob[\,|OPT(\mat)| = r\,] \geq \frac{1}{2}.
    \]
    Selecting sufficiently large $k$, we have $(1 - p)^{(r+1)^2} \geq (1 - \epsilon)$,  obtaining the statement of  Theorem \ref{thm:badconcentration}.
        
\end{proof}
\end{document}